\documentclass[a4paper,11pt]{article}

\usepackage{fullpage}
\usepackage[margin=1in]{geometry}
\usepackage{relsize}
\usepackage{parskip}

\usepackage{geometry}
\usepackage{amsmath, amsthm, amssymb, amstext, comment, graphicx}
\usepackage{subfigure}

\usepackage{cite}
\usepackage{amsmath,amssymb,amsfonts}
\usepackage{algorithmic}
\usepackage{textcomp}
\usepackage{xcolor}

\definecolor{darkbrown}{rgb}{0.4, 0.26, 0.13}

\newtheorem{theorem}{Theorem}[section]
\newtheorem{definition}[theorem]{Definition}

\newtheorem{corollary}[theorem]{Corollary}

\newtheorem{lemma}[theorem]{Lemma}

\usepackage{hyperref}

\hypersetup{
	colorlinks   = true, 
	urlcolor     = blue, 
	linkcolor    = blue, 
	citecolor   = blue 
}

\def\asymmetric{0}
\def\symmetric{1}

\def\uniform{0}

\newcommand{\zmax}{{z_{\max}}}
\newcommand{\zmin}{{z_{\min}}}

\newcommand{\tNL}{\theta_{NL}}  
\newcommand{\tNB}{\theta_{NB}}  
\newcommand{\tLB}{\theta_{LB}}  

\begin{document}

\title{How to Charge Lightning:\\The Economics of Bitcoin Transaction Channels\footnote{The authors are in alphabetical order. An earlier version of this paper was presented at the Scaling Bitcoin 2017 workshop: \url{https://stanford2017.scalingbitcoin.org/}.} }
\author{Simina Br\^anzei\thanks{Purdue University, USA. E-mail: \href{mailto:simina.branzei@gmail.com}{simina.branzei@gmail.com}.} \and Erel Segal-Halevi\thanks{Ariel University, Israel. E-mail: \href{mailto:erelsgl@gmail.com}{erelsgl@gmail.com}.} \and Aviv Zohar\thanks{Hebrew University of Jerusalem, Israel. E-mail: \href{mailto:avivz@cs.huji.ac.il}{avivz@cs.huji.ac.il}.}}

\maketitle

\begin{abstract}
Off-chain transaction channels represent one of the leading techniques to scale the transaction throughput in cryptocurrencies. However, the economic effect of transaction channels on the system has not been explored much until now. 

We study the economics of Bitcoin transaction channels, and present a framework for an economic analysis of the lightning network and its effect on transaction fees on the blockchain. Our framework allows us to reason about different patterns of demand for transactions and different topologies of the lightning network, and to derive the resulting fees for transacting both on and off the blockchain. 

Our initial results indicate that while the lightning network does allow for a substantially higher number of transactions to pass through the system, it does not necessarily provide higher fees to miners, and as a result may in fact lead to lower participation in mining within the system.
\end{abstract}


\section{Introduction}
A main approach to solve the scalability problem in Bitcoin is to use off-chain transaction channels that allow parties to transfer funds while communicating directly, and only occasionally to settle on the blockchain. The  deployment of SegWit, a solution to transaction malleability (among other benefits) opens the path for better constructions of off-chain transaction channels. While transaction channels themselves are limited to exchanges between pairs of individuals, further developments like the lightning network \cite{PD16} allow to route payments over longer paths and thus can allow the construction of a well connected network of payment channels that can be used to transfer money quickly and with relatively little interaction with the blockchain. For additional discussion of micropayment channels and scalability, see, e.g., \cite{HS17,DW15,Croman2016}.

One of the key unknowns regarding fast payment networks is the economic effect that they will have on the Bitcoin fee market. If the blockchain is used less often, fees to miners are paid less frequently and competition for space in blocks declines. Bitcoin's security depends heavily on having a large amount of computational power invested in solving proof-of-work puzzles, making it hard for attackers to double spend or censor transactions in the currency. As the block reward in Bitcoin declines (halving every four years), the reliance on fees increases and these must suffice to pay for enough mining by honest participants. 

In this work we explore the economics of Bitcoin transaction channels, and in particular the economic equilibrium that results from the introduction of fast-payment networks to the Bitcoin ecosystem. Our main contribution is a theoretical framework in which one can reason about the usage of payment channels and the cost of committing records to the blockchain. We explore different topologies of payment channels and find the market equilibrium that dictates (among other things) the fees that will be collected by miners, the transactions that pass through the lightning network or directly through the blockchain, and the transactions that do not take place (e.g., micro transactions for which the fees are too high on both alternatives). 
\medskip 

\noindent \textbf{Summary of our findings.}
While our findings strongly depend on the assumptions we have made regarding the distribution of payment sizes and willingness to pay fees, we generally find that the revenue obtained by miners can sometimes be lower when lightning networks are deployed compared to when they are not, unless extremely large numbers of participants take part in transacting (there too, results depend on the distributional assumptions). Naturally, the addition of payment channels does indeed result in very high transaction throughputs in the system overall.

The implications for Bitcoin are that 
the revenue from fee payments alone might be insufficient to support the security of the system. Still, we caution that our results should be taken with a grain of salt: we account for very simple topologies but hope our initial exploration will inspire further exploration of similar models. 

\subsection{An overview of our approach}
The first step taken when modeling the effects of different transaction methods is to select a model for the demand for transactions --- a model specifying which participants want to send money, to whom, what amount is transfered, and how much fee the sender is willing to pay to complete the transfer. 

A second aspect that needs to be determined is the topology of payment channels that is set in the system: which pairs of participants choose to establish channels between them, and how much funding is dedicated to each such channel. This is crucial in determining the lifetime of the channel given different use patterns. 

Finally, with this information at hand, we set out to compute the demand for blockchain records. Such demand stems from two main sources: the establishment and settlement of existing lightning channels, and direct transfers that occur on the blockchain. Given that the block size is limited, the daily supply of new transaction records is fixed, and we are thus able to compute the market equilibrium fee for records.

\medskip 

\noindent \textbf{The mechanics of transaction channels.}
Transaction channels are typically established by locking funds using a single blockchain transaction. The channel state is then updated by the two participants by exchanging transactions that update the division of funds from the locked amount. These transactions are not typically transmitted to the blockchain. Each update represents a new division of funds and usually only the last transaction is committed to finalize the transfer. The transactions exchanged by the two participants on each channel are set up so that if one of the participants (say Alice) disappears or tries to take funds that are not hers (e.g., by placing a transaction that represents an old state on the blockchain), the other participant (say Bob) can recover his funds or even punish Alice by taking all the funds in the channel. For the purpose of this work, we assume all channels are established and settled cooperatively, as we aim to consider the expected behavior of the market under ``normal'' circumstances. 
	
We additionally assume that channels can be settled and reopened with a single transaction, which may be of larger size in the network.
We assume the parties that transfer money back and forth do so according to a random process, and that they therefore occasionally end up in a state where all funds in the channel are directed to one of the users. In this case, the channel can only be used to transfer money in one direction and the channel must be reset or re-funded to allow flow in the other direction. Clearly, if the amount of liquidity in the channel is high, then this event will occur rarely. It is therefore of paramount importance to establish the typical amount of funding in each channel. Since liquidity that is locked within the channel represents money that is not invested elsewhere, we consider the cost of holding liquidity in the channel as the lost income from interest payments on this sum. Stated differently: we allow participants to borrow as much money as they want to fund their transaction channels, and the cost for such payments is simply the interest rate in the economy. This cost is the de facto limiting factor for the lightning network. 

\medskip 

\noindent \textbf{Models for the demand for transactions.}
We explore two primary models for the demand for transactions. The first is a model in which participants are paired and only transact with their direct partner. While this is not a realistic depiction of the flow of money in an economy, it is in some sense a best-case setting for transaction channels, as no routing of payments is required by the system. The topology of channels in this case is also simple: just create channels between transacting pairs if it is profitable to do so.  In this setting we note there are several variants: one in which transactions occur with equal probability in each direction, and one in which transfers are  asymmetrically biased in one direction. In this work we focus on the symmetric variant. Another axis along which we vary our analysis is regarding the size of payments. We assume in one case that payments come from a uniform distribution, and in another case that they are derived from a power-law distribution (as it  often is in real life data). 

Our second model assumes that all participants may pay each other, and that payments occur between participants that are chosen uniformly at random. Here we focus on an analysis of a payment network that includes a single payment hub. The hub needs to maintain additional liquidity, but allows participants that are connected to it to route payments to everyone else. We find surprisingly that the results bear great similarity to the pairs model (except for extra payments for the additional liquidity). 

\medskip 
\noindent \textbf{Implementation.} The simulation code can be downloaded from: \href{https://github.com/erelsgl/bitcoin-simulations}{https://github.com/\\erelsgl/bitcoin-simulations}.


\subsection{Related Work}
Following the circulation of our paper, there have been follow-up works investigating the lightning network.  \cite{wang2022can} analyze the structure of the lightning network and compare it with the Barabasi–Albert model
 for generating random scale-free graphs via preferential attachment, finding that the lightning network has a different structure (such as different assortativity and  diameter). \cite{ZFDS22} conduct an empirical analyis of the lightning network, focusing on the  betweenness centrality distribution of the routing system.

\cite{BCV20} studies the lightning network as a  percolation process and aims to understand how the distribution of volume and size of transactions impacts the feasability of the system.    \cite{Beres2021Cryptoeconomic} design  a publicly available traffic simulator to empirically study the transaction fees and privacy provisions on the lightning network. 
  \cite{GGS21}  analyze payments on lightning channels that are unidirectional or symmetric bidirectional. Their work identifies conditions for two parties to optimally establish a channel, find explicit formulas for channel costs and obtain the optimal collaterals and savings entailed, deriving the implied reduction in the congestion on the blockchain.

 Lightning channels may be subject to adversarial attacks, which may allow an adversary to discover channel balances, thus threatening the privacy of the users. \cite{NFSD20}  observe that  the lightning network allow users to use gossip and probing mechanisms to learn about possible paths for routing their transactions, which may in turn be exploited by an adversary to learn information about the transactions. In this context, \cite{NFSD20} analyze two types of attacks: a probing attack, where an adversary wants to detect the maximum amount transferable in a given direction on a channel by probing it, and a timing attack, where the adversary discovers how close the destination of a routed payment is.
 \cite{BNT22} analyze probing attacks in the presence of multiple channels between the same pair of users on the lightning network. \cite{Kappos_21} consider multiple types of attacks that an adversary may perform in the lightning network, including balance, path, and payment discovery.

\section{Model}

We analyze the market for \emph{records} on the blockchain. A record is a part of a block, in which a single transaction is recorded. Each record has a market-price $\phi$ [bitcoins-per-record], which is the mining-fee for a blockchain transaction. The market-price $\phi$ is determined as a price in which the \emph{supply} of records equals the \emph{demand} for records. 

The \emph{supply} in our market is quite simple: the bitcoin protocol ensures that the supply of records per day is fixed. We denote this parameter by $\tau$. 
The total revenue of the miners, which is an important factor in the security of bitcoin, will be $\tau\cdot \phi$ [bitcoins-per-day].

The \emph{demand} is driven by the need of users to transfer money to other users. The demand is determined by the following sets of parameters:
\begin{itemize}
\item The number of times that user $i$ wants to transfer money to user $j$ per day is a Poisson random variable with mean value $\lambda_{i,j}$ [transfers-per-day].
 
\item The size of transfers from user $i$ to user $j$ is $z_{i,j}$ [bitcoins-per-transfer]. This models the fact that some users do micro-transfers while others do bulk transfers.
We will sometimes assume that $z_{i,j}$ is drawn from a probability distribution such as uniform or power-law.

We assume that the transfer-size $z_{i,j}$ is constant for each pair, i.e, it is drawn randomly once for each pair and then remains fixed.

\item The utility that user $i$ gains from each transfer to user $j$ is $v_{i,j}$ [bitcoins-per-transfer]. 
We will often assume that the utility is proportional to the transfer size, i.e, $v_{i,j} = \beta z_{i,j}$, where $\beta$ is a constant  in $(0,1)$.

\end{itemize}

Each time user $i$ considers transferring money to user $j$, it compares three options: blockchain transfer, lightning transfer, or no transfer at all. 
The user selects the option with the highest net gain. 
In the case of a blockchain transfer, the net gain is $v_{i,j}$ minus the blockchain fee $\phi$. 
In the case of a lightning transfer, the net gain is $v_{i,j}$ minus the \emph{lightning fee}, which is derived from the cost of maintaining the lightning channel.
In case of no transfer, the net gain is zero. 

The lightning fee is derived from several parameters which determine the cost of using lightning:
\begin{itemize}
\item The number of blockchain records required for a channel-reset transaction, denoted by $a$. Note 
a reset transaction is slightly larger than a standard transaction, so $a$ is a number between $1$ and $2$ [records]. Therefore the cost of resetting a lightning channel is $a\cdot \phi$ [bitcoins].
\item 
The (fixed) interest rate $r$ [per day]. 
A user who wants to use lightning has to lock money in channels, and thus it has to pay an economic cost determined by $r$. This means that in general a user will not want to lock all its money in lightning channels; the user will look for the optimal amount to lock such that the total cost (economic cost plus channel-reset cost) is minimized. 
\end{itemize}
\medskip 

We study several special cases for the transfer matrix $\lambda_{i,j}$:
\begin{itemize}
\item \emph{Pairs}: the users are divided in pairs (e.g. $(1,2), (3,4), \ldots$) All transfers are only inside each pair, i.e, for every $i$, only $\lambda_{2i,2i-1}$ and $\lambda_{2i-1,2i}$ are non-zero. This is in some sense the best case for lightning, since we need a channel only between user $i$ and user $j$. 
\item \emph{Symmetric uniform}: $\lambda_{i,j}=\lambda$ for all $i$ and $j$. 
\item \emph{Asymmetric uniform}: for each pair $i,j$, $\lambda_{i,j}+\lambda_{j,i} = \ell$ and $|\lambda_{i,j}-\lambda_{j,i}|=\Delta$. I.e, the pairs are asymmetric, either user $i$ transfers more money to user $j$ or vice-versa.
\end{itemize}
In general, it is possible that some agents accumulate money endlessly while other agents spend money endlessly; this can be explained by assuming that they do some transfers outside bitcoin. Alternatively, one can assume a special topology in which all nodes have an even degree and have positive net transfer in exactly half their edges.

\smallskip 

We also consider several possible lightning topologies:
\begin{description}
\item[\emph{Pairs}:] each user $2i$ has a channel only with user $2i-1$.
\item[\emph{Star}:] there is a single node (``bank'') which is connected to all users; all transfers go through the bank.
\end{description}

\section{Analysis of a Single Channel}
The basic building-block for our analyses is the analysis of a single channel between two users, Alice and Bob. 
We first analyze the expected channel lifetime for a given channel capacity and distribution of transfers, and then find the optimal channel capacity.

We checked two cases related to the transfer rates: in the symmetric case the transfer-rate from Alice to Bob equals the transfer-rate from Bob to Alice; in the asymmetric case the transfer-rates are different.
\if\asymmetric0
For brevity, this paper presents only the results for the symmetric case; we found the results for the asymmetric case to be qualitatively similar. Thus from now on we assume the channels are symmetric.
\fi

\if\asymmetric1
\subsection{Lifetime of Asymmetric Channels}
We start by analyzing the expected channel lifetime when in each round there is a single unit (coin) sent from Alice to Bob with probability $p$ and from Bob to Alice with probability $1-p$.

\begin{theorem}[Lifetime of Asymmetric Channel] \label{thm:asymmetric_life}
Let Alice and Bob have a channel with $w$ bitcoins, which they use to transfer single bitcoins to each other at a time. Each transfer is made by Alice with probability $p$ and by Bob with probability $1-p$. If $p \in (0,0.5) \cup (0.5,1)$,
then the expected lifetime of the channel from the state where Alice has $m$ bitcoins and Bob has $w - m$ bitcoins is:
\[
X_m = \frac{m}{2 p-1} - \left( \frac{w}{2 p-1} \right) \cdot \frac{1 - \left( \frac{p}{1-p} \right)^m}{1 - \left(\frac{p}{1-p} \right)^w}
\]
\end{theorem}
\begin{proof}
Let $X_m$ be the expected lifetime of the channel, in number of transfers, from the state where Alice has $m$ bitcoins and Bob has $w - m$. The channel lasts until one of the players reaches a balance of zero bitcoins.
Then $X_m = 1 + p \cdot X_{m-1} + (1-p) \cdot X_{m+1}$, which gives the recurrence
$$X_{m+1} = \left( \frac{1}{1-p}\right) \cdot X_m - \left(\frac{p}{1-p}\right) \cdot X_{m-1} - \frac{1}{1-p}$$

Note that $X_{0} = 0$, since if Alice has zero in her account and wants to make a transfer, then the channel is reset. Similarly, $X_{w} = 0$ corresponds to the case where Bob has zero in his account.

The associated homogeneous recurrence relation is $X_{m+1}^h = \left( \frac{1}{1-p}\right) \cdot X_m^h - \left(\frac{p}{1-p}\right) \cdot X_{m-1}^h$, with 
characteristic equation $x^2 - \left( \frac{1}{1-p} \right) x + \frac{p}{1-p} = 0$. The roots are $x_1 = \frac{p}{1-p}$ and $x_2 = 1$. 

The general solution has the following form, for some constants $a,b$:
\begin{equation}  \label{eq:general2}
X_m^h = a + b \left( \frac{p}{1-p} \right)^m
\end{equation}

We must additionally find a particular solution. 
We write $X_{m+1} = \left( \frac{1}{1-p}\right) X_m - \left(\frac{p}{1-p}\right)  X_{m-1} + F(m+1)$, where $F(m) = - \frac{1}{1-p}$ for all $m$.
In general, when the non-homogeneous term in such a recurrence is 
$$
F(m) = \left(\sum_{k=0}^{t} b_k \cdot m^k \right)  s^m,
$$
if $s$ is a root with multiplicity $q$ of the characteristic equation, there exists a general solution to the recurrence such that for some constants $p_0 \ldots p_t$:
\begin{equation} \label{eq:general}
X_m^p = m^q  \left( \sum_{k=0}^t p_k \cdot m^k \right)  s^m
\end{equation}

Since $F(m) = -\frac{1}{1-p}$, we obtain $s=1$, $t=0$, $b_0 = \frac{-1}{1-p}$, and $b_i = 0$ for all $i \in \mathbb{N}^*$. Note $s=1$ is a solution of order $1$ to the characteristic equation, thus $q=1$. Replacing $s$ and $q$ in Equation \ref{eq:general}, the particular solution has the form 
$$
X_m^p = m \left( \sum_{k=0}^t p_k \cdot m^k \right),
$$
In this case, the particular solution has the form $X_m^p = um$, for some constant $u$.
Substituting in the recurrence relation, it follows that 
\begin{eqnarray*}
X_{m+1}^p & = & \left( \frac{1}{1-p} \right) X_m^p - \left( \frac{p}{1-p} \right) X_{m-1}^p - \frac{1}{1-p} \iff \\ 
u \cdot (m+1)& = & \left(\frac{1}{1-p} \right) u \cdot m - \left( \frac{p}{1-p} \right) u \cdot (m-1) - \frac{1}{1-p}
\end{eqnarray*}
We obtain $u = \frac{1}{2 p-1}$, which implies our particular solution is 
\begin{equation} \label{eq:particular}
X_m^p =  \frac{m}{2 p-1}
\end{equation}
Adding the homogeneous and non-homogeneous parts (Equations \ref{eq:general2} and \ref{eq:particular}), we get 
$$
X_m = X_m^p + X_m^h = a + b \left( \frac{p}{1-p} \right)^m + \frac{m}{2 p-1}
$$
From the boundary condition $X_{0} = 0$, which is equivalent to the channel closing since Alice has no money left, it follows that 
$$
a + b \left( \frac{p}{1-p} \right)^0 + \frac{0}{2 p-1} = 0 \iff b = - a 
$$
The boundary condition $X_{w} = 0$ gives
$$ 
a - a\left( \frac{p}{1-p}\right)^{w} + \frac{w}{2 p-1} = 0 \iff a = \frac{\frac{w}{2 p-1}}{\left(\frac{p}{1-p} \right)^w - 1} 
$$
It follows that for all $m = 0 \ldots w$, the expected number of transfers made on the channel is:
\begin{align}
\label{eq:lifetime-asymmetric-p}
X_m = \frac{m}{2 p-1} - \left( \frac{w}{2 p-1} \right) \cdot \frac{1 - \left( \frac{p}{1-p} \right)^m}{1 - \left(\frac{p}{1-p} \right)^w}  \,.
\end{align}
This completes the proof for asymmetric channels.
\end{proof}

\begin{corollary}
Suppose there is a lightning channel between Alice and Bob, with a capacity of $w$ bitcoins. Alice and Bob transfer money through the channel, a single coin each time.     The transfers from Alice to Bob are a Poisson process with mean $\lambda_A$ transfers-per-day, and the transfers from Bob to Alice are a Poisson process with mean $\lambda_B $ transfers-per-day.
If $\lambda_B \neq \lambda_A$, then the expected lifetime of the channel from the state where Alice has $m$ bitcoins and Bob has $w - m$ bitcoins, is:

\begin{align*}
\tilde{X}_m = \frac{m}{\lambda_A-\lambda_B} - \frac{w}{\lambda_A-\lambda_B}  \left( \frac{1 - \left( \frac{\lambda_A}{\lambda_B} \right)^m}{1-\left(\frac{\lambda_A}{\lambda_B} \right)^w} \right)
&& \text{[days]}
\end{align*}

\end{corollary}
\begin{proof}
This follows from a well-known property of Poisson processes: given two such independent processes with means $\lambda_A$ and $\lambda_B$, the probability that the next event will be from the first process is $\frac{\lambda_A}{\lambda_A + \lambda_B}$ and the probability that the next event will be from the second process is $\frac{\lambda_B}{\lambda_A + \lambda_B}$.
Denote $p := \frac{\lambda_A}{\lambda_A + \lambda_B}$. 
Then, $2 p - 1 = \frac{\lambda_A-\lambda_B}{\lambda_A + \lambda_B}$ and $\frac{p}{1-p} = \frac{\lambda_A}{\lambda_B}$.
Then, the channel behaves exactly as in Theorem \ref{thm:asymmetric_life} so its lifetime, in number of transfers, is:
\begin{align*}
X_m =
\frac{m(\lambda_A+\lambda_B)}{\lambda_A-\lambda_B} - \left( \frac{w(\lambda_A+\lambda_B)}{\lambda_A-\lambda_B} \right) \cdot \frac{1 - \left( \frac{\lambda_A}{\lambda_B} \right)^m}{1 - \left(\frac{\lambda_A}{\lambda_B} \right)^w}
&& \text{[transfers]}
\end{align*}
To get the lifetime in number of days, 
we just have to divide the above expression by the mean number of transfers per day, which is $\lambda_A+\lambda_B$. This completes the proof.
\end{proof}

Note that the left-hand term can be seen as the net drift from Alice to Bob, while the right-hand term as stochastic ``noise''.

\begin{figure}
\begin{center}
\includegraphics[scale=0.4]{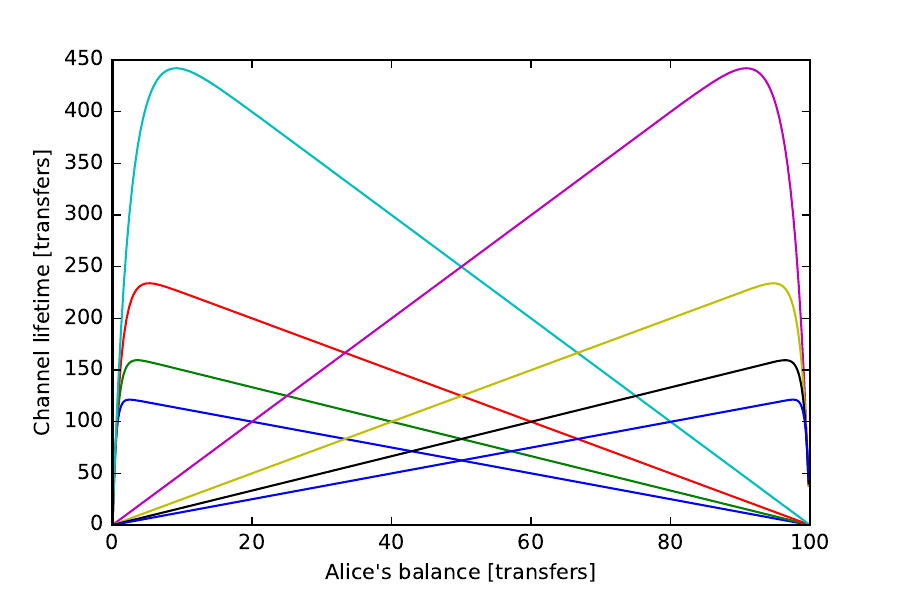}
\end{center}
\caption{
\label{fig:channel-lifetime}
Lifetime of a channel between two users as given by equation (\ref{eq:lifetime-asymmetric-p}).
The channel capacity is $100$ and the lifetime is given in number of transfers.
In the right, from top to bottom, $p$ is $0.4, 0.3, 0.2, 0.1$. In the left, from top to bottom, $p$ is $0.6, 0.7, 0.8, 0.9$.
\footnotesize{
Plot credit: 
Matplotlib \citep{matplotlib}.
}
}
\end{figure}

Figure \ref{fig:channel-lifetime} shows some plots of the channel lifetime for different values of $p$. The curves are essentially linear in the interior of the interval. Therefore, from now on we will work with a linear approximation of the channel lifetime for asymmetric channels.

\begin{definition}[Linear Approximation of Asymmetric Channel Lifetime]
Consider a channel between Alice and Bob with $w$ bitcoins, in which Alice transfers to Bob using a Poisson process with mean $\lambda_A$ and Bob to Alice using a Poisson process with mean $\lambda_B$. If $\lambda_A \neq \lambda_B$, the lifetime of a channel can be approximated by the following linear term:
\begin{equation*}
\tilde{X}_m = \frac{m}{|\lambda_A - \lambda_B|}
\end{equation*}
\end{definition}
\fi  

\if\symmetric1

\subsection{Channel Lifetime}

\begin{theorem}[Lifetime of Symmetric Channel] \label{thm:symmetric_life}
	Let Alice and Bob have a channel with $w$ bitcoins, which they use to transfer single bitcoins to each other at a time. Each transfer consists of one coin sent by Alice with probability $1/2$ and by Bob with  probability $1/2$.
	Then the expected lifetime of the channel from the state where Alice has $m$ bitcoins and Bob has $w - m$ bitcoins is
$X_m = w m - m^2$.
\end{theorem}
\begin{proof}
	Let $X_m$ denote the expected channel lifetime from the state where Alice has $m$ bitcoins and Bob has $w-m$. The channel lasts until one of the players reaches a balance of zero bitcoins.

	Then $X_m = 1 + 1/2  \cdot X_{m-1} + 1/2 \cdot X_{m+1}$, which gives the recurrence
	$X_{m+1} = 2 \cdot X_m - X_{m-1} - 2$ (see, e.g., \cite{recurrence_lecture_notes}).
	
We have $X_{0} = 0$, since if Alice has zero in her account and wants to make a transfer, then the channel is reset. Similarly, $X_{w} = 0$ corresponds to the case where Bob has zero in his account.
	
	
	Then the associated homogeneous recurrence relation is $X_{m+1}^h =2 \cdot X_m^h - X_{m-1}^h$, with 
	characteristic equation $x^2 - 2 x + 1 = 0$.
The roots of the associated characteristic equation are $x_1 = x_2 = 1$, and so the general solution has the following form, where $a,b$ are constants:
\begin{equation} \label{eq:general_phalf}
X_m^h = a + bm \,.
\end{equation}

In general, if $s$ is a root with multiplicity $q$ of the characteristic equation, there exists a general solution to the recurrence such that for some constants $p_0 \ldots p_t$:
\begin{equation} \label{eq:sym_general}
X_m^p = m^q  \left( \sum_{k=0}^t p_k \cdot m^k \right)  s^m
\end{equation}

The particular solution, given by Equation \ref{eq:sym_general}, has $q=2$, $s=1$, and $t=0$, so
$
X_m^p = vm^2
$ for some constant $v$. Then
$$
v \cdot (m+1)^2 =2 v \cdot m^2 - v \cdot (m-1)^2 - 2 \iff v = -1 \,.
$$
The particular solution is  $X_m^p = -m^2$.
Adding the homogeneous part (\ref{eq:general_phalf}) and the non-homogeneous part (\ref{eq:sym_general}) gives:
$$
X_m = X_m^h + X_m^p = a + bm - m^2\,.
$$
The boundary condition $X_0 = 0$ gives $a = 0$, while $X_w = 0$ gives $b = w$. Then for all $n = 0 \ldots w$, the expected number of transfers made on the channel is 
$
X_m = w m -m^2
$.
\end{proof}

\begin{corollary}
	Let Alice and Bob have a channel with $w$ bitcoins, which they use to transfer single bitcoins to each other at a time.
	Suppose both users transfer to each other using a Poisson process with mean $\lambda_A = \lambda_B = \lambda$.
	Then the expected number of days for the channel started from the state where Alice has $m$ bitcoins and Bob has $w - m$ bitcoins is: 
\begin{align}
\tilde{X}_m = \frac{w m - m^2}{2\lambda} &&\text{[days]} \notag 
\end{align}
\end{corollary}
\begin{proof}
When the players send to each other with equal means, Alice sends to Bob a single unit with probability $1/2$ and Bob sends to Alice with probability $1/2$.
From Theorem \ref{thm:symmetric_life}, the expected lifetime from the state where Alice has $m$ bitcoins is $X_m = w m - m^2$.
To obtain a solution in days, we divide this by the mean number of transfers per day, which is $2 \lambda$. 
\end{proof}
\fi 

\subsection{Channel Optimal Initialization}
Given the expressions for the channel lifetime, Alice and Bob can calculate the optimal way to initialize a channel with a fixed capacity $w$, i.e. the initial balance $(a,b)$ that maximizes the expected lifetime of the channel.

\if\symmetric1
\begin{lemma}
	The optimal initialization of a symmetric channel with capacity $w$ is for both Alice and Bob to start with $w/2$ bitcoins.
\end{lemma}
\begin{proof}
The expected channel lifetime from the state where Alice has $m$ bitcoins is ${wm - m^2}/{(2\lambda)}$; the lifetime is maximized when $m = w/2$. An optimally initialized symmetric channel is expected to last for:
$T(w) = {w^2 /(8\lambda)}$.
\end{proof}
\fi 

\if\asymmetric1
\begin{lemma}
The (approximately) optimal initialization of an asymmetric channel is to give all the funds to the player with the largest flow.
\end{lemma}
\begin{proof}
In this case we solve the linear approximation, which is equivalent to a deterministic flow from Alice to Bob if $\lambda_A > \lambda_B$ and from Bob to Alice if $\lambda_B > \lambda_A$. Then it is immediate that the optimal way to initialize the channel is to give all the money to the player with the larger flow. The (approximate) lifetime of the optimally initialized channel is:
$\tilde{T}(w) = {w\over  |\lambda_A- \lambda_B|}$.
\end{proof}
\fi 

So far, we assumed that each transfer between 
Alice and Bob consists of a single coin.
However, it is easy to generalize the results
to a situation in which each transfer between Alice and Bob consists of $z$ bitcoins, for some constant $z$. 
In this case, the channel funding $w$ is given in \emph{transfers}, i.e, the channel lifetime is $T(w)$ when the channel is funded with $z\cdot w$ \emph{bitcoins}.

\subsection{Channel Optimal Capacity}
Alice and Bob have two options for performing a sequence of transactions, namely sending money to each other on the blockchain or through a channel in the lightning network. In the case of blockchain transactions, there will be a fixed fee $\phi$ paid for each transaction, where $\phi$ is equal to the price (in bitcoins) of a blockchain record. 

If on the other hand Alice and Bob decide to use a lightning channel, they will incur an interest $r$ on the money that is locked into the channel, in addition to the fee $a\cdot \phi$ paid when the channel is reset and recorded on the blockchain. 
We assume that this cost is covered by charging a fixed fee per lightning-transaction. The fee is calculated such that the expected fee charged during the entire channel lifetime equals the total cost of maintaining the channel.

\begin{lemma} \label{lem:expected_fee}
	Let $\phi$ be the blockchain mining fee and $r$ the daily interest rate for locking money in a lightning channel. 
	Suppose the transfers between Alice and Bob are Poisson processes with means $\lambda_A,\lambda_B$ transfers-per-day, respectively, with 
	$\ell := \lambda_A + \lambda_B$.
	Each transfer is of $z$ bitcoins,
	and Alice and Bob use 
	a lightning channel with capacity $w\cdot z$ bitcoins.
	
	Then, to cover the maintenance costs of the channel, the fee per transaction should be:
	\begin{align*}
	{F(w) = \frac{w z \cdot (1+r)^{T(w)} +a \phi - w z}{T(w)\cdot \ell}} && \text{\emph{[bitcoins / transaction]}}
	\end{align*}
	where $T(w)$ is the expected channel lifetime (in days).
\end{lemma}
\begin{proof}
The players pay interest on the quantity $w z$ locked in the channel until the channel closes, i.e. for $T(w)$ days in expectation. Since the interest is paid each day (rather than for each transaction), the cost due to the interest is:
$$
\left((1+r)^{T(w)} - 1 \right) \cdot w z \,.
$$
The users incur an additional cost of $a \phi$ at the end for recording the channel balance on the blockchain. 

The expected number of transactions during the channel lifetime is $T(w)\cdot \ell$, so the expected fee charged during the channel lifetime is $F(w)\cdot T(w)\cdot \ell$. This fee should cover the costs, so we should have:
\begin{align*}
\left((1+r)^{T(w)} - 1 \right) \cdot w z + a\cdot w 
= 
F(w)\cdot T(w)\cdot \ell,
\end{align*}
from which the claim follows.
\end{proof}

\begin{lemma}[First-order approximation]
	\label{lem:foa}
	Consider the function that gives the expected fee per transaction:
	$$F(w) =
	\frac{
		\left((1+r)^{T(w)} - 1 \right) \cdot w z + a \phi
	}
	{
		T(w)\cdot \ell
	}\,.$$
		The first-order approximation of $F$ 
 taken around $r = 0$ is
	\begin{align*}
	\tilde{F}(w) 
	= 
	\frac{w z r}{\ell} + \frac{a \phi}{(T(w)\cdot \ell)},
	\end{align*}
	where $T(w)$ is the expected channel lifetime and is independent of the interest rate $r$.
\end{lemma}
\begin{proof}
	Let $g(r) = (1+r)^{T(w)}$, where $r \in \Re^{+}$. 
	Then 
	\begin{align*}
	F(w) =
	\frac{
		\left(g(r) - 1 \right) \cdot w z + a \phi
	}
	{
		T(w)\cdot \ell
	}
	\end{align*}
	Since $T(w)$ is independent of $r$, the first derivative of $g$ is: $$g'(r) =
	T(w) \cdot (1+r)^{T(w)-1}\,.$$
	Then the first order approximation of $g(r)$ around $r_0 = 0$ is:
	\begin{align*}
	\tilde{g}(r) = & \; g(0) + g'(0) \cdot (r - r_0) 
	= 1 + {r \cdot T(w)}.
	\end{align*}
	Substituting $\tilde{g}(r)$ for $g(r)$ in $F(w)$, we get: 
	\begin{align*}
	\tilde{F}(w) 
	= & 
	\frac{w z r T(w) + a \phi}{T(w)\cdot \ell} \,.
	\end{align*}
	The error term for the estimation of $g(r)$ is given by a function $h(r)$ with the property that $|h(r)| \leq \frac{M r^2}{2}$, where $M$ is the maximum value taken by $|g''(r)|$ on the interval $[0,c]$. We have: $$g''(r) = T(w)\left(T(w) - 1 \right) \cdot (1+r)^{T(w)-2}\,.$$
	There are a few cases:
	\begin{itemize}
		\item $x \geq 2 \sqrt{2 \ell}$: $M_1 = \max_{r \in [0,c]} |g''(r)|$ is attained for $r = 1$, so we have that $$M_1 = \left|T(w) \cdot \left( T(w) - 1 \right) \cdot (1+c)^{T(w) - 2} \right|\,.$$
		\item $x < 2 \sqrt{2\ell}$: the maximum is $M_2 = \left|T(w) \left( T(w) - 1 \right) \right|$, attained for $r = 0$.
	\end{itemize} 
	Take $M = \max\{M_1, M_2\}$. 
		Let $R(w) = F(w) - \tilde{F}(w)$ be the error term in the approximation of $F(w)$. Then 
	\begin{align}
	|R(w)| = \left| \frac{w z \cdot (g(r) - \tilde{g}(r))}{T(w)\cdot \ell} \right| 
	= \frac{w z \cdot \left|h(r) \right|}{T(w)\cdot \ell}  \leq \frac{2Mr^2}{x} \, .\notag 
	\end{align}
\end{proof}

The cost of a transaction on the blockchain is $\phi$, thus performing transactions on the lightning network is profitable for Alice and Bob as long as $F \leq \phi$.

\if\symmetric1
\begin{lemma}[First order approximation of symmetric channel capacity] \label{lem:optimal-capacity}
	Let $\phi$ be the blockchain mining fee and $r$ the daily interest rate for locking money in a lightning channel. 
	Suppose the transfers between Alice and Bob are Poisson processes with means $\lambda_A = \lambda_B = \lambda$ transfers-per-day, with 
	$\ell := \lambda_A + \lambda_B = 2\lambda$.
	Each transfer is of $z$ bitcoins,
	and Alice and Bob use 
	a lightning channel with capacity $w\cdot z$ bitcoins.
	
	Then the first order approximation of the expected cost per transaction on the channel is $$\frac{z r w}{\ell} + \frac{4 a \phi}{w^2}\,.$$
	The optimal channel capacity based on this approximation is $w_{opt} = \left({\frac{8 a \phi \ell}{z r}}\right)^{1/3}$.
\end{lemma}
\begin{proof}
	From Lemma \ref{lem:foa}, the 
	first-order approximation of the lightning-fee $F$ is:
	$
	\tilde{F}(w) 
	= 
	{w z r}/{\ell} + {a \phi}/{(T(w)\cdot \ell)}\,.
	$
	
	For a symmetric channel, $T(w) = w^2/4 \ell$ [days], thus 
	$$
	\tilde{F}(w) = 
	\frac{z r w}{\ell} + \frac{4 a \phi}{w^2},
	$$
	as claimed. The error term is given by:
	$$
	|R(w)| 
	= \frac{w z}{T(w)\cdot \ell} \cdot \left|h(r) \right|
	= \frac{4 z}{w} \cdot \left|h(r) \right|
	\leq \frac{2Mr^2}{x}\,.
	$$
	
	Minimizing the fee given the interest rate $r$, blockchain fee $\phi$, and mean number $\ell$ of daily transfers, is equivalent to minimizing the function 
	\begin{align*}
	f(x) =  \frac{4 z x \cdot (1+r)^{\frac{x^2}{4\ell}} + 4 a \phi - 4 z x}{x^2}, 
	\end{align*}
	{where} $x \in [1, \infty)$ is the channel capacity [in transfers].
	
	To find the optimal channel funding for a given $r$, $\phi$, $\ell$, we must find the minimum of $f(x)$ on the interval $[1, \infty)$. 
	As explained before, the first-order approximation of $f$ is:
	\begin{align*}
	\tilde{f}(x) 
	= \frac{z r x}{\ell} + \frac{4 a \phi}{x^2}\,.
	\end{align*}
	
	The optimal channel funding based on this approximation can be determined by taking the minimum of the function $\tilde{f}(x)$ on the range $(0, \infty)$. Note that $$\tilde{f}'(x) = \frac{z r}{\ell} - \frac{8 a \phi}{x^3} \; \; \; \mbox{and} \; \; \; \tilde{f}''(x) = \frac{24 a \phi}{x^4} > 0\,.$$ Thus $\tilde{f}$ is convex and the minimum is attained for $x_{min}$ with the property that $\tilde{f}'(x_{min}) = 0$, that is, 
	$
	x_{min} = \left(\frac{8 a \phi \ell}{z r}\right)^{1/3}
	$.

	The lightning fee given the optimal funding is:
	\begin{align*}
	F^{opt} = F(x_{\min}) = 
	3\left({\frac{a \phi z^2 r^2}{\ell^2}}\right)^{1/3} \,.
	\end{align*}
		This completes the proof.
\end{proof}
\fi 

\if\asymmetric1
\begin{lemma}[First order approximation of asymmetric channel capacity]
	Let $\phi$ be the blockchain mining fee and $r$ the daily interest rate for locking money in a lightning channel. 
	Suppose the transfers between Alice and Bob are Poisson processes with means $\lambda_A \neq \lambda_B$ transfers-per-day, with 
	$\ell := \lambda_A + \lambda_B$
	and $\Delta := |\lambda_A - \lambda_B|$.
	Each transfer is of $z$ bitcoins,
	and they use 
	a lightning channel with capacity $w\cdot z$ bitcoins.
	
	Then, the first order approximation of the expected cost per transaction on the channel is $\frac{w z r}{\ell} + \frac{a \phi \Delta}{w\cdot \ell}$. The optimal channel capacity based on this approximation is $w_{opt} = \sqrt{\frac{a \phi \Delta}{z r}}$.
\end{lemma}
\begin{proof}
	In the asymmetric case, the channel lifetime is $T(w) = w / \Delta$ [days]. Substituting this in Lemma \ref{lem:foa} gives:
	\begin{align*}
	\tilde{F}(w) 
	= & 
	\frac{w z r}{\ell} + \frac{a \phi}{T(w)\cdot \ell}
	=
	\frac{w z r}{\ell} + \frac{a \phi \Delta}{w\cdot \ell}
	\end{align*}
	By taking the derivative by $w$ we get that the optimal funding is:
	\begin{align*}
	w_{opt} = \sqrt{\frac{a \phi \Delta}{z r}}.
	\end{align*}
	as claimed.
	Then the lightning fee given the optimal funding is:
	\begin{align*}
	F^{opt} = \tilde{F}(w_{opt}) = 
	2\sqrt{
		\frac{a \phi \Delta z r}{\ell^2}
	}
	\end{align*}
\end{proof}
\fi 

\section{Pairs Topology}
In this section we assume that the
network of transfers has a very simple topology: the users are divided to pairs who only trade with each other.

Our goal is to calculate the market-equilibrium blockchain-fee $\phi$.
To this end, we have to calculate the demand for each value of $\phi$ and find the value of $\phi$ in which the demand equals the supply (which we assume is constant, $\tau$ records per day). 

\subsection{Symbolic computations}
We start by analyzing the demand of a single pair.
\begin{lemma}[Choice of a single pair]
\label{lem:choice}
Let $\phi$ be the blockchain mining fee.
Suppose the transfers between Alice and Bob are Poisson processes with means $\lambda_A, \lambda_B$ transfers-per-day respectively, where 
$\ell := \lambda_A + \lambda_B$, $\Delta:=|\lambda_A - \lambda_B|$, and each transfer is of $z$ bitcoins. 
Then there exist thresholds $\tNL,\tNB,\tLB$ such that Alice and Bob:
\begin{itemize}
	\item do not transfer to each other when $z < \min\{ \tNL,\tNB \}$.
	\item  transfer only via lightning when $\tNL < z < \tLB$.
	\item  transfer only via  blockchain when $z > \max\{\tNB,\tLB \}$.
\end{itemize}
The thresholds depend on the transfer rates.
\if\symmetric1
With symmetric transfer rates, they are: 
\begin{align} 
\tNL  = {27 a r^2 \over \ell^2\beta^3} \cdot \phi; \; \; \; 
\tNB = {\phi \over \beta}; \; \; \; 
\tLB  = {\ell\over \sqrt{27 a r^2}} \cdot \phi\,. \notag 
\end{align} 
\fi 
\if\asymmetric1
With asymmetric transfer rates: $\tNL = {4\Delta a r \over \ell^2\beta^2} \phi$, $\tNB = {\phi\over \beta}$, and $\tLB = {\ell^2\over 4\Delta a r} \phi$.
\else  
\fi   
\end{lemma}

\begin{proof}
The users consider three options: (1) do not transfer at all, (2) transfer through the lightning channel between them, or (3) use direct blockchain transfers.
They will choose the option with the 
lowest net utility. We assume that the users have quasi-linear utilities, so the utility of each user is the value of the transfer minus the fee paid.

The net utility of doing no transfer is clearly $0$.
If the value of a transfer from Alice to Bob is $v = \beta\cdot z$, then the net utility of doing a blockchain transfer is:
$u_{B} =
{v - \phi}
=
\beta z - \phi\,.$
The net utility of a lightning transfer is $v$ minus the lightning fee calculated in the previous subsection, assuming optimal channel funding: 
$$u_{L} = \beta z - F^{opt} = \beta z - 
3\left({a \phi z^2 r^2 \over \ell^2}\right)^{{1}/{3}}\,.$$

The possible relations between the three net utilities 
\if\asymmetric1
are given by:
\begin{align*}
u_L > 0 && \iff  && z > {4\Delta a r \over \ell^2\beta^2} \phi
&= \tNL
\\
u_B > 0 && \iff  && z > {1\over \beta} \phi
&= \tNB
\\
u_B > u_L && \iff  && z > {\ell^2\over 4\Delta a  r} \phi
&= \tLB
\end{align*}
\fi 
\if\symmetric1
 are:
 \begin{itemize}
     \item $u_L > 0  \iff   z > 
\left(\frac{27 a r^2}{\ell^2\beta^3}\right) \cdot \phi
= \tNL$.
\medskip 
\item $u_B > 0 \iff   z > \frac{1}{\beta} \cdot \phi
= \tNB$.
\medskip 
\item $u_B > u_L  \iff  z > 
\left(\frac{\ell}{\sqrt{27 a r^2}} \right) \cdot \phi
= \tLB$.
\end{itemize}
\medskip 
\fi 

The choice of the users depends on these parameters as follows.
If $0 > \max(u_L,u_B)$, then the users choose no-transfer.
If $u_B > \max(0,u_L)$, then the users transact via the blockchain.
If $u_L > \max(0,u_B)$, then the users transact via their lightning channel.
The claim follows from the above by straightforward calculations.
\end{proof}

It is noteworthy that all three thresholds are linear functions of the blockchain-fee $\phi$, even though 
the channel-lifetime is a non-linear function of $\phi$.

We ignore the case in which two of the three utilities are exactly equal  (e.g, $u_L = 0$ or $u_B = 0$ or $u_L = u_B$), since  the transfer size $z$ is drawn from a continuous distribution, so the probability of an exact equality is zero.

The previous lemma assumed that the transfer-size $z$ is fixed. Now, we let $z$ be a random variable, and calculate the expected demand per pair. 
We calculate the demand function both with and without lightning.
\begin{lemma}[Demand of a single pair]
\label{lem:demand}
Let $\phi$ be the blockchain mining fee.
Suppose the transfers between Alice and Bob are Poisson processes with means $\lambda_A, \lambda_B$ transfers-per-day respectively, with 
$\ell := \lambda_A + \lambda_B$ and $\Delta:=|\lambda_A - \lambda_B|$,
and the transfer-size $z$ is a random variable with probability-distribution-function $f$.

Let $\tNL,\tNB,\tLB$ be as described in Lemma~\ref{lem:choice}.
Then the expected demand of the two users, measured in records-per-day, is:
\begin{itemize}
\item Without lightning: $D^{0}(\phi) =
\ell
\int\limits_{\tNB}^{\infty}  f(z) dz$.
\item With lightning: 
\[
D^s(\phi) = 
\int\limits_{\tNL}^{\tLB} 
{\left({\ell a r^2 {z}^2 \over \phi^2}\right)^{\frac{1}{3}}}
f(z) dz
+
\ell   
\left(\int\limits_{\max(\tNB,\tLB)}^{\infty} 
f(z) dz\right)
\]
\end{itemize}



The transactions count and volume are as follows:

\begin{table}[h!] 
	\label{tab:summary}
	\begin{center}
		\begin{tabular}{||c | c|c|c|c ||} 
			\hline \hline
			 & Lightning & Blockchain  \\
			\hline \hline
			Transactions count & $\ell  
\int\limits_{\tNL}^{\tLB} 
f(z) dz$ & $\ell  
 \int\limits_{\max(\tNB,\tLB)}^{\infty} 
f(z) dz$ \\
			\hline
			Transactions volume & $\ell  
\int\limits_{\tNL}^{\tLB} 
z f(z) dz$  & $\ell  
\int\limits_{\max(\tNB,\tLB)}^{\infty} z f(z) dz$   \\
			\hline \hline 
		\end{tabular}
	\end{center}
\end{table}




\medskip 


\vspace{-3mm}
\end{lemma}
\begin{proof}
By Lemma \ref{lem:choice}, when $z > \max(\tNB,\tLB)$, the users use blockchain transfers. They do $\ell$ transfers per day, each of which uses a single record, so their demand is $\ell$ records-per-day. This explains the rightmost terms in all expressions (i.e. for $D^0(\phi)$ and $D^s(\phi)$).

When $\tNL < z < \tLB$, the users transact via their lightning channel, so their demand is 
$a$ records each $T$ days, where $T$ is the channel-lifetime (in days) when it is funded optimally. So their daily demand is 
\if\asymmetric1
$a/T = {\sqrt{\Delta a r z \over \phi}}$ in the asymmetric case, and
\fi
\if\symmetric1
$\left({\ell a r^2 {z}^2 / \phi^2}\right)^{1/3}$. 
\fi
This explains the leftmost term in the expression for $D^s(\phi)$, completing the proof.
\end{proof}

We also analyze  examples with  realistic numbers in Appendix~\ref{sub:numeric-example} and observe that usually the thresholds are ordered such that $\tNL \ll \tNB \ll \tLB$.

\medskip 

Now we can calculate the equilibrium blockchain-fee. 

\begin{theorem}[Equilibrium blockchain-fee]
\label{lem:equilibrium-fee}
Suppose there are $n$
users who only interact in pairs (i.e. there are $n/2$ pairs of users).
In every pair, all parameters ($\lambda_A, \lambda_B, \ell$) are the same, 
except for the transfer-size $z$, which is drawn randomly for each pair from probability-distribution $f$.
Then, the market-price $\phi_{eq}$ is a 
solution to the  equation:
${n\over 2}\cdot D(\phi) = \tau,$
where $D(\phi)$ is given by the expression in Lemma \ref{lem:demand} for each case. 
\end{theorem}
\begin{proof}
Assuming all pairs have the same parameters, the expected aggregate demand is ${n\over 2}\cdot D(\phi)$ records-per-day. The supply is $\tau$ records-per-day. The market-equilibrium price is the price in which the demand equals the supply.
\end{proof}

\subsection{Power-law distribution}
We now present a special case of Lemma \ref{lem:demand} for the case in which the transfer-size $z$ is distributed according to the following power-law:
\begin{align*}
f(z) = 
\begin{cases}
0 & \mbox{ if } z < \zmin
\\
{\zmin/ z^2} & \mbox{ if }  z \geq \zmin\,.
\end{cases}
\end{align*}
For simplicity we assume that $\tNB < \tLB$. 


\begin{corollary}
\label{cor:demand-powerlaw}
When the transfer-size has power-law distribution and $\tNB < \tLB$, the expected demand of the pair (measured in records-per-day) is as follows in each case:
\begin{itemize}
    \item Without lightning: 
\begin{align*}
D^{0}(\phi) & =
\begin{cases}
\ell & \mbox{\emph{ if }} \;\; \tNB < \zmin
\\
\ell \cdot {\zmin\over \tNB} & \mbox{\emph{ if }} \;\; \tNB \geq \zmin \,.
\end{cases}
\end{align*}
\end{itemize}
\begin{itemize}
    \item With lightning: 
\begin{align*}
D^s(\phi) &= 
\begin{cases}
\ell, \mbox{\emph{ if }}\;\; \tLB < \zmin
\\ \\
3\zmin\left({l a r^2\over \phi^2}\right)^{\frac{1}{3}} \left[\frac{1}{\sqrt[3]{\zmin}} - \frac{1}{\sqrt[3]{\tLB}}\right] + \ell \cdot {\zmin\over \tLB},  \\ \;\;\; \mbox{\emph{ if }}\;\; \tNL < \zmin \leq \tLB 
\\  \\
3\zmin {\left(l a r^2\over \phi^2\right)^{\frac{1}{3}}} \left[\frac{1}{\sqrt[3]{\tNL}} - \frac{1}{\sqrt[3]{\tLB}}\right] + \ell \cdot {\zmin\over \tLB}, \\ \;\;\; \mbox{\emph{ if }}\;\; \tNL \geq \zmin \,.
\end{cases}
\end{align*}
\end{itemize}

The transaction counts on lightning and blockchain, respectively, are as follows:
%
\begin{table}[h!] 
	\label{tab:summary_transaction_counts}
	\begin{center}
		\begin{tabular}{||c | c|c|c|c ||} 
			\hline \hline
			 Transaction counts & Lightning & Blockchain  \\
			\hline \hline
			$\tLB < \zmin$ & $0$ & $\ell$  \\
			\hline
			$\tNL < \zmin \leq \tLB$ & $\ell \left(1-{\zmin\over{\tLB}}\right)$  & $\ell \left({\zmin\over{\tLB}} \right)$   \\
			\hline \hline 
			\hline
			$\tNL \geq \zmin$ & $\ell \left({\zmin\over\tNL}-{\zmin\over{\tLB}} \right)$ & $\ell \left({\zmin\over{\tLB}} \right)$   \\
			\hline \hline 
		\end{tabular}
	\end{center}
\end{table}

 
The expected transaction volumes are  as follows:

\begin{table}[h!] 
	\label{tab:summary_transaction_volumes}
	\begin{center}
		\begin{tabular}{||c | c|c|c|c ||} 
			\hline \hline
			 Transaction volumes & Lightning & Blockchain  \\
			\hline \hline
			$\tLB < \zmin$ & $0$ & $\infty$  \\
			\hline
			$\tNL < \zmin \leq \tLB$ & $\ell \zmin \Bigl(\ln{\tLB} - \ln{\zmin}\Bigr)$  & $\infty$   \\
			\hline \hline 
			\hline
			$\tNL \geq \zmin$ & $\ell \zmin \Bigl(\ln{\tLB} - \ln{\tNL}\Bigr)$ & $\infty$   \\
			\hline \hline 
		\end{tabular}
	\end{center}
\end{table}

%
%
\end{corollary}

\section{Simulations}
In the theoretical analysis we assumed that the transfer-size $z$ is randomized once per pair of agents. 
We found  that the case in which the transfer-size is randomized in each transfer is much more difficult to analyze theoretically. Therefore we study this case using simulations. 

\subsection{Channel operation and definitions}
We consider a channel with a total capacity $w$. Since the transfer-rate from Alice to Bob is the same as from Bob to Alice, we assume that the channel is initialized symmetrically - each agent contributes $w/2$. The "channel state" is the balance for  Alice in the channel; therefore the initial state is $w/2$.

Whenever a transfer has to be made, we determine its size $z$ by randomly drawing from the power-law distribution defined by $f(z) = {\zmin / z^2} \; [z>\zmin]$.
Then, we check whether the transfer can be done in the current channel state (i.e, whether the agent making the transfer has a sufficiently high balance). If the transfer can be done on the channel, then it is done and the channel state is updated. Otherwise, we check whether it is worthwhile to do the transfer on the blockchain: if the transfer value $\beta z$ is larger than the current $\phi$, then the transfer is done on the blockchain; otherwise, it is not done at all.

In the initial state ($w/2$), the chances of channel failure are relatively small; as the channel drifts away from this initial state towards one of its endpoints, the chances of channel failure increase. Therefore, at some states it may be worth doing a \emph{channel reset} and returning the channel to its initial state. We assume  there is a constant $R$ such that, whenever the state after a channel transfer is smaller than $R$ or larger than $w-R$, it is reset to $w/2$. We call $R$ the \emph{reset radius}. 

\subsection{Finding the optimal reset-radius}
The first step in analyzing the channel performance is finding the optimal reset radius. The plots in Figure \ref{fig:sim-reset-radius}.a illustrate the channel performance as a function of the reset radius, for a capacity of $w=10$.

\begin{figure}[h!]
    \centering
    \includegraphics[scale=0.65]{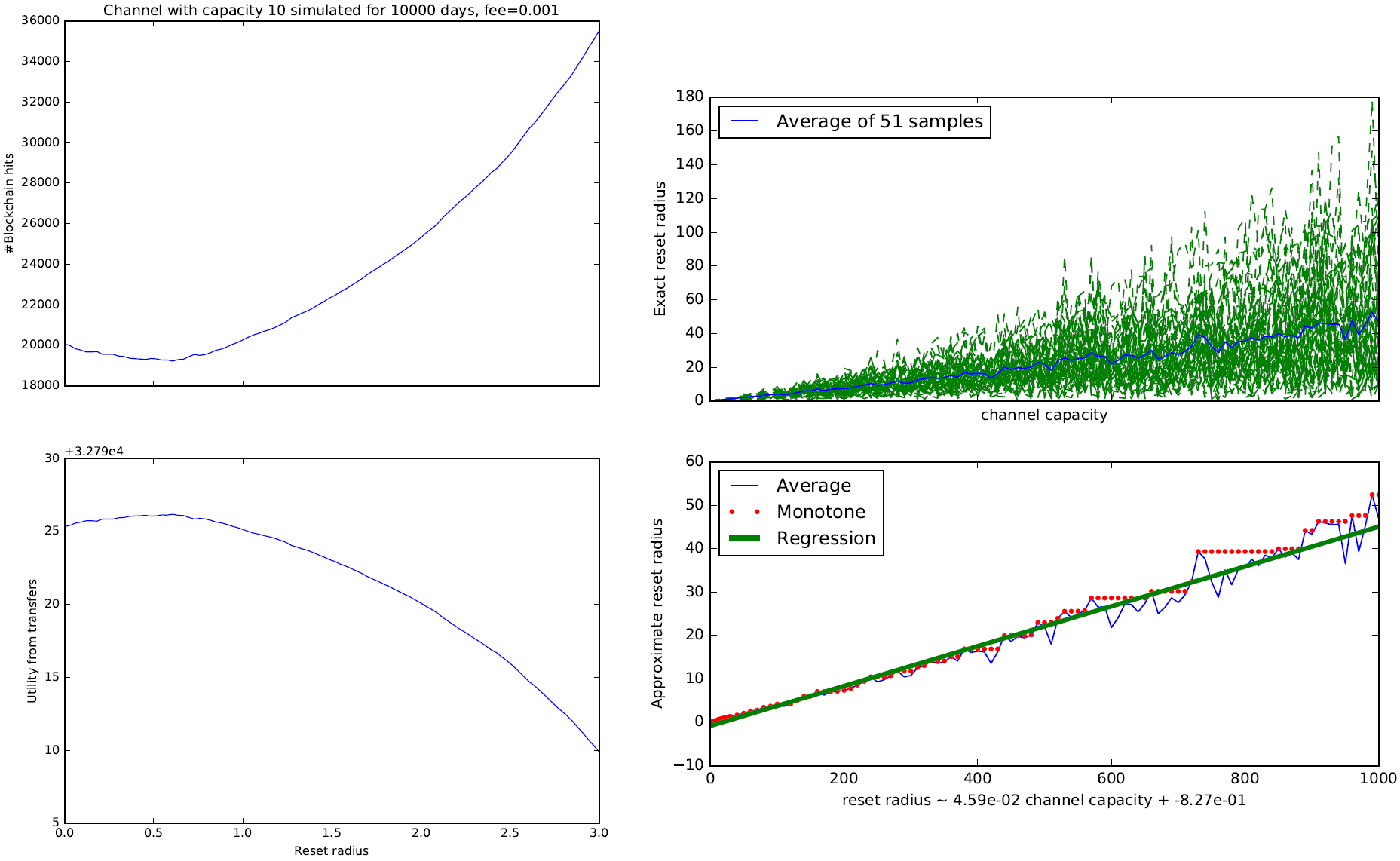}
    \caption{\textcolor{darkbrown}{\em Left}: channel performance as a function of  reset radius.
\textcolor{darkbrown}{\em Right}: optimal reset radius as a function of  channel capacity. For the right figure, the text at the bottom represents the linear regression formula. The linear regression shows that, approximately, $y$ = 4.59e-02 $x$ - 8.27e-01, where $y$ is the reset-radius and $x$ is the channel-capacity.}
\label{fig:sim-reset-radius}
\end{figure}

The first plot shows the number of \emph{blockchain hits} --- the number of blockchain records required for a typical sequence of transfers (this is the sum of the records used for out-of-channel transfers and the records used for reset transactions). The number of blockchain hits is minimized when the reset radius is about $0.5$.
The second plot shows the net utility of the agents, which is their value from making the transfers ($\beta$ times the transfer sizes) minus the fee paid for blockchain hits. The utility is maximized at approximately the same point, $0.5$. 
The third plot shows how many transfers are done on the blockchain and how many are done on the channel, as a function of the reset radius.

We now calculate the optimal reset radius. Naturally, the optimum depends on the capacity $w$ and on the blockchain fee $\phi$. 
For simplicity we calculate and plot the optimal reset radius only as a function of $w$, for a fixed fee of $0.001$. We run 50 experiments, each of which simulates transfers over 1000 days. 
Figure \ref{fig:sim-reset-radius}/Right
shows the average optimal radius, as well as a linear regression line; it shows that the optimal reset radius is approximately 5 percent of the channel capacity.

\subsection{Finding the optimal channel capacity}
From now on, we assume that each channel with capacity $w$ is operated with the optimal reset-radius for this capacity, using the linear regression formula shown in Figure \ref{fig:sim-reset-radius}/Right.
Our next goal is to calculate the optimal channel capacity, given the blockchain fee (and the fixed interest rate $r$).


\begin{figure}
\begin{center}
\includegraphics[scale=0.66]{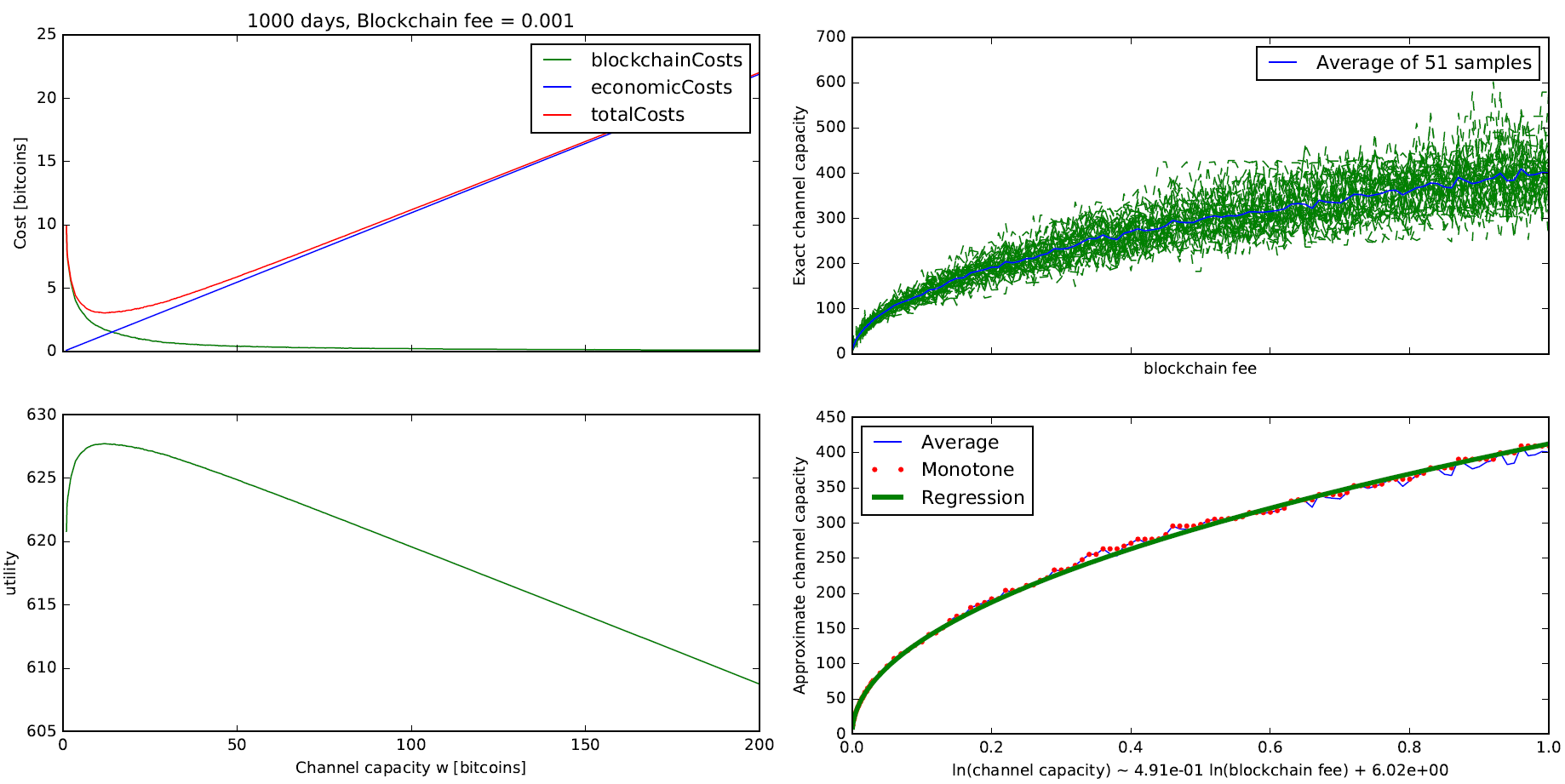}
\end{center}
\caption{\textcolor{darkbrown}{\em Left}: channel performance as a function of the channel capacity.
\textcolor{darkbrown}{\em Right}: optimal channel capacity as a function of the blockchain fee.}
\label{fig:sim-channel-capacity}

\end{figure}

Figure \ref{fig:sim-channel-capacity}/Left shows some plots of the channel performance as a function of the channel capacity, for different blockchain fees. As $w$ increases, the \emph{blockchain cost} (the cost from blockchain transfers and reset transactions) decreases, but the \emph{economic cost} (i.e. the cost from interest) increases. The optimal utility is attained at a capacity in which these two costs are approximately equal.

We now calculate the optimal channel capacity as a function of the blockchain fee. We run 50 experiments, each of which simulates transfers over 1000 days. Figure \ref{fig:sim-channel-capacity}/Right shows the average channel capacity, as well as a curve calculated by log-log regression; it shows that the optimal channel capacity is approximately a constant times $\phi^{1/2}$.
The theoretical analysis in Lemma \ref{lem:optimal-capacity} gives an estimate of  $O(\phi^{1/3})$, which is qualitatively similar (recall that the situations are different since here we randomize the transfer-size each transfer rather than once per pair).

\subsection{Demand curves}
From now on, we assume that each channel is funded with the optimal $w$ given the blockchain fee $\phi$, and operated with the optimal reset-radius given that $w$.

We calculate the curve of demand for blockchain records as follows. 
For each value of $\phi$ we calculate the number of records required to do the randomly-drawn transfers. This calculation is done in two steps: first, we calculate the total cost for doing all these transfers in a lightning channel (the blockchain cost plus the economic cost). We assume that this cost should be covered by setting a fee on each lightning transaction. The lightning fee is proportional to the transfer size and the value of a transfer is also proportional to the transfer size. Thus, if the total cost is smaller than the total value of all transfers, then all transfers will be done on the lightning channel. If the total cost is larger than the total value, then no transfers will be done on the lightning channel. In the latter case, the users will use only  blockchain, and they will do only the transfers with a value above the blockchain fee.

\begin{figure}[h!]
\begin{center}
\includegraphics[scale=0.95]{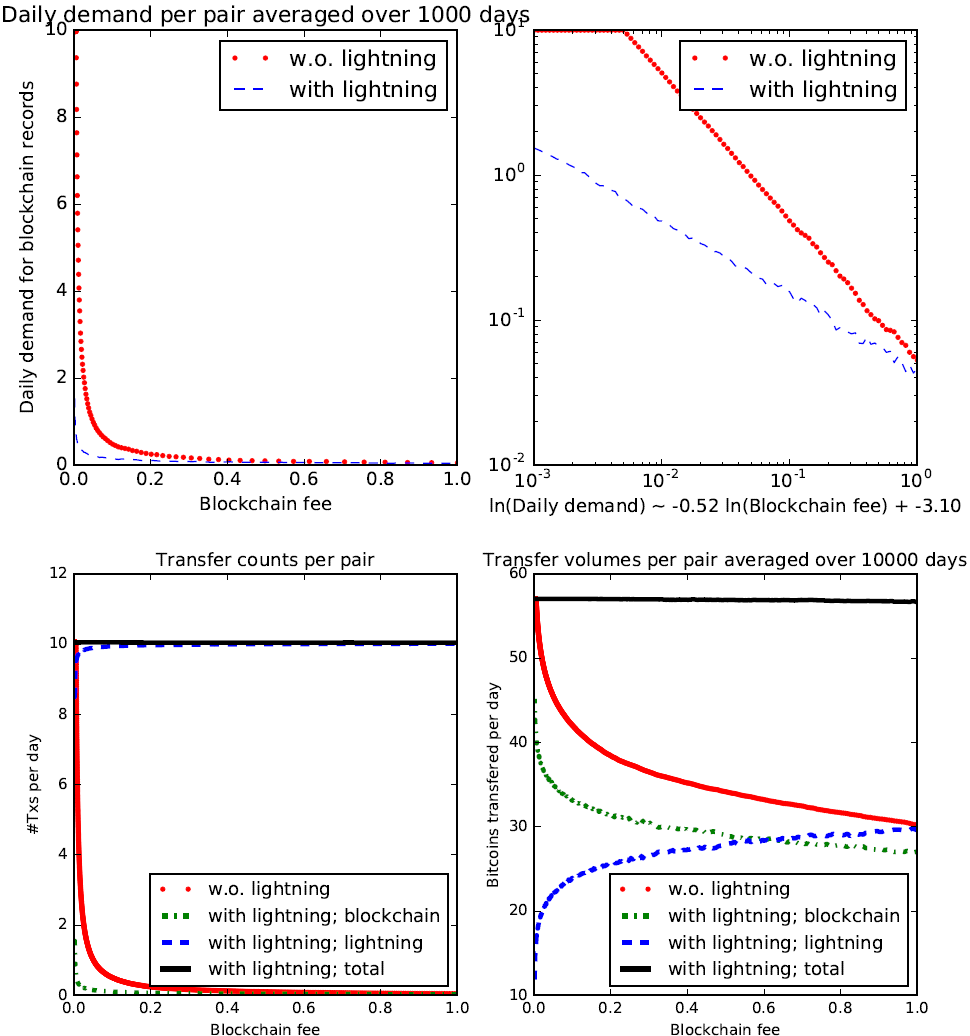}
\end{center}
\caption{
\label{fig:sim-demand-curves}
\textbf{\textcolor{darkbrown}{Top:}} Demand curves with and without lightning.
\emph{Left}: linear-linear plots.
\emph{Right}: log-log plots, and regression formula for the lightning curve.
\\
\textbf{\textcolor{darkbrown}{Bottom}}: transactions with and without lightning.
\emph{Left}: transaction-counts.
\emph{Right}: transaction volumes.
}
\end{figure}

Figure \ref{fig:sim-demand-curves}/Top shows some plots of the daily demand for blockchain records (per pair of users) as a function of the blockchain fee. Log-log regression shows that the demand is with lightning is approximately $\Theta(1/\phi^{1/2})$ and the demand without lightning is approximately $\Theta(1/\phi)$.

Figure \ref{fig:sim-demand-curves}/Bottom shows the division of transactions between lightning and blockchain. Interestingly, while almost all transfers are done on the lightning network (Left), most transfer-volume is done on the blockchain (Right).

\subsection{Equilibrium fee}
Our next goal is to calculate the price-curves, describing the equilibrium blockchain fee as a function of the number of users. Given the number of users $n$, we estimate the equilibrium fee $\phi$ in the following way:
\begin{itemize}
\item Draw a random sequence of transfer-sizes for 1000 days (an average of $1000 \ell = 10000$ transfers).
\item Define the \emph{demand function} $D(\phi)$ as the number of blockchain-hits when the above sequence is performed in a lightning channel with optimal capacity and optimal reset-radius.
\item Define the \emph{surplus function} $s(\phi)$ as the difference between the demand and the supply per day, i.e: $s(\phi) := {n\over 2}\cdot D(\phi)/1000 - \tau\,.$
\item Find $\phi^*$ for which $s(\phi^*)=0$ (we used scipy.optimize.brentq, which is based on a method of \cite{brent}  to find zeroes).
\end{itemize}

We ran the above procedure for 100 values of $n$ between $10^5$ and $10^8$ (logarithmically spaced). For each value, we repeated the random calculation 500 times. We calculated both the average $\phi$ for each $n$ (averaged over the 500 samples), and a regression curve of $\phi(n)$. The results are shown in Figure \ref{fig:sim-price-curves}.

\begin{figure}[h!]
\begin{center}
\includegraphics[scale=0.98]{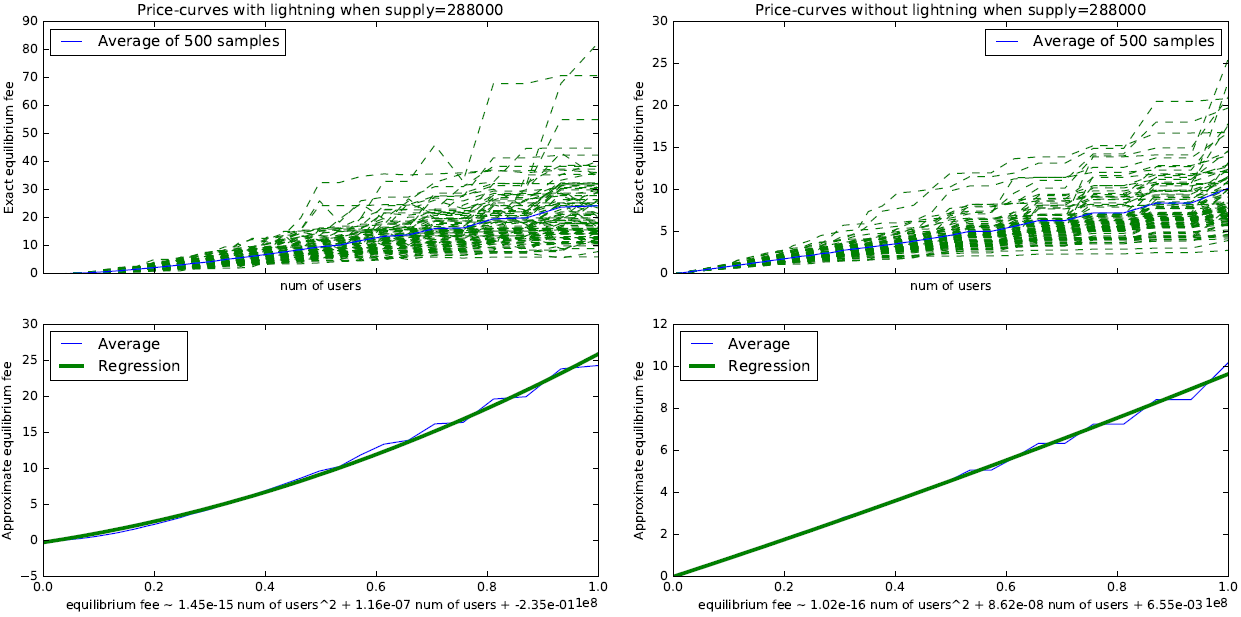}
\end{center}
\caption{
\label{fig:sim-price-curves}
\textbf{\textcolor{darkbrown}{Left}}: Price curves with lightning.
\textbf{\textcolor{darkbrown}{Right}}: Price curves with no lightning.
}
\end{figure}

The blockchain fee grows super-linearly with the number of users. The data exhibits a very good fit with a polynomial of degree 2, which indicates that the blockchain fee grows like $\Theta(n^2)$, where $n$ is the number of users.
Without lightning, the price grows linearly with $n$. 

\subsection{Overall network performance}
Finally, we simulated the overall network performance as a function of the number of users. We simulated random transfers over 100,000 days, and 1000 values of $n$ between $10^5$ and $10^8$ (logarithmically spaced). For each value, we calculated the equilibrium fee with and without lightning, the number and the volume of transfers done in blockchain vs. lightning. 
We did this experiment once for a  block size of 288000 records per day  and another time for a  block size twice as large (of 576000 records per day). The results are shown in Figure \ref{fig:sim-net-stats}.

The results are qualitatively similar with single or double block size. As the number of users grows, the equilibrium fee grows super-linearly (first line) and with it, the miners' revenue (fifth line). However, the utility per user decreases (sixth line). All the supply of records is used (second line). Almost all transfers are done via lightning (third line), but the volume of transfers is split approximately equally between lightning and blockchain (fourth line), since the largest transfers are done on the blockchain.

\begin{figure}[h!]
\begin{center}
\includegraphics[scale=0.85]{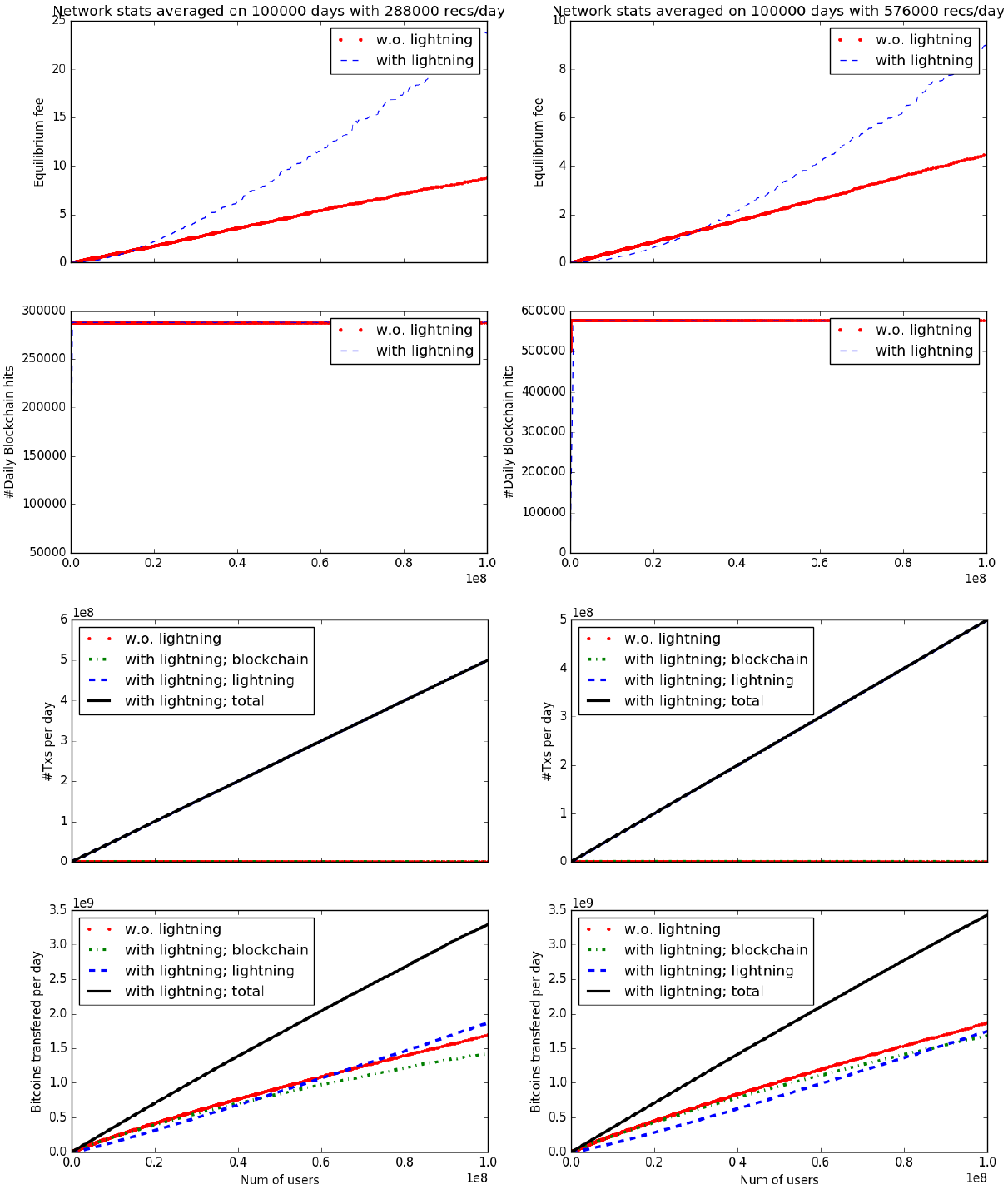}
\end{center}
\caption{
\label{fig:sim-net-stats}
\textbf{\textcolor{darkbrown}{Right}}: Network stats with  block size of 288000 records per day.
\textbf{\textcolor{darkbrown}{Left}}: Network stats with  block size twice as large (576000 records per day).
}
\end{figure}


Doubling the block size (with lightning) has some minor quantitative differences: the blockchain fee is smaller by a factor of about 3, the miners' revenue is smaller by a factor of about 1.5, and the utility per user is higher by a factor of 1.5 (see Figure~\ref{fig:sim-net-stats-utility-and-revenue}).

\begin{figure}[h!]
\centering 
\includegraphics[scale=1.2]{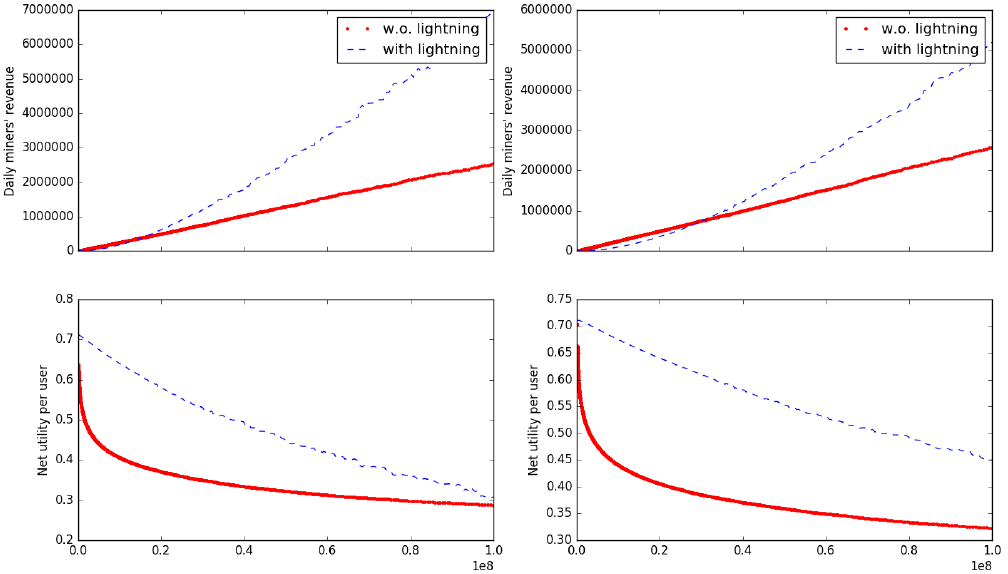}
	\caption{
		\label{fig:sim-net-stats-utility-and-revenue}
		\textbf{\textcolor{darkbrown}{Right}}: Network stats with  block size of 288000 records per day.
		\textbf{\textcolor{darkbrown}{Left}}: Network stats with  block size twice as large (576000 records per day).
	}
\end{figure}


\section{Star Topology}
We now show how to solve networks where there is a central bank that each user is connected to.
For each pair of users $i,j$, let $\lambda_{i,j}$ denote the flow of user $i$ towards user $j$. 
Then we can define $\lambda_{i}^{+} = \sum_{j \in N} \lambda_{i,j}$ and $\lambda_{i}^{-} = \sum_{j \in N} \lambda_{j,i}$. Set $p_i = \frac{\lambda_i^+}{\lambda_i^- + \lambda_i^+}$ and $q_i = \frac{\lambda_i^-}{\lambda_i^- + \lambda_i^+}$. Then $p_i + q_i = 1$, so the expected lifetime of the channel between user $i$ and the bank can be obtained as a corollary from the single channel analysis.

\begin{corollary}[Lifetime of Channel between User and Bank]
Suppose there is a lightning star network, where the transfers from each user $i$ to any other user $j$ are given by a Poisson process with mean $\lambda_{i,j}$ transfers per day, and the channel between user $i$ and the bank has capacity $w$.
Let $\lambda_{i}^{+} = \sum_{j \in N} \lambda_{i,j}$ and $\lambda_{i}^{-} = \sum_{j \in N} \lambda_{j,i}$.
Then, the expected lifetime of the channel 
between user $i$ and the bank from the state where user $i$ has $m$ bitcoins and the bank has $w-m$ bitcoins on the channel with user $i$ is:

(a) when $\lambda_{i}^+ \neq \lambda_i^-$:
	\begin{align*}
	\tilde{X}_m^i = \frac{m}{\lambda_{i}^+-\lambda_{i}^-} - \frac{w}{\lambda_{i}^+-\lambda_{i}^-}  \left( \frac{1 - \left( \frac{\lambda_{i}^+}{\lambda_{i}^-} \right)^m}{1-\left(\frac{\lambda_{i}^+}{\lambda_{i}^-} \right)^w} \right)
	&& \text{[days]}
	\end{align*}
	
(b) when $\lambda_{i}^+ = \lambda_i^-$: 
\begin{align*}
\tilde{X}_m^i = {w m - m^2 \over 2\lambda_i^+} && \text{[days]}\,.
\end{align*}
\end{corollary}

\begin{corollary}
The optimal initialization of a star network with channels of capacity $w$ is such that for each user $i$
\begin{itemize}
	\item If $\lambda_i^+ = \lambda_i^-$, both the bank and user $i$ start with $w/2$ bitcoins. The expected lifetime of user $i$'s channel with the bank when initialized this way is $T_i(w) = {w^2 \over 8\lambda_i^+}$.
	\item If $\lambda_i^+ > \lambda_i^-$, user $i$ starts with $w$ bitcoins and the bank with $0$ bitcoins, while if $\lambda_i^+ < \lambda_i^-$ user $i$ starts with $0$ bitcoins and the bank with $w$ bitcoins. The approximate expected lifetime of user $i$'s channel is $\tilde{T}_i(w) = {w\over  |\lambda_i^+- \lambda_i^-|}$.
\end{itemize}
\end{corollary}

\medskip

Next we calculate the optimal channel funding for the star topology, finding that the fee per transaction for any user in the star network is twice as high compared to the pairs topology.

\medskip

\begin{theorem} \label{lem:expected_fee_star}
	Consider a star network with users $1 \ldots n$ connected through a bank, such that the transfer from any user $i$ to another user $j$ is a Poisson process with mean $\lambda_{i,j}$,
	with $\ell_i = \sum_{k \in N} \lambda_{i,k} + \lambda_{k,i}$.
	Let $\phi$ be the blockchain mining fee and $r$ the daily interest rate for locking money in a lightning channel. 
	Each transfer between any pair of users is of $z$ bitcoins, and each user $i$ has a channel with the bank of capacity $w_i \cdot z$ bitcoins.
	
	Then the expected cost of any transfer that a user $i$ is involved in is:
	\begin{align*}
		{F(w_i) = \frac{w_i z \cdot (1+r)^{T(w_i)} +a \phi - w_i z}{T(w_i)\cdot \ell_i}} && \text{[bitcoins/transfer]}
	\end{align*}
	where $T(w)$ is the expected lifetime (in days) of a channel of capacity $w_i$.
		The expected fee paid by each user $i$ is twice as high in the star network compared with the model where user $i$ transacts only with a fixed user $j$ without a bank. 
\end{theorem}
\begin{proof}
	This follows from the single channel analysis (Lemma \ref{lem:expected_fee}), by noting that the expected (total) cost of a transaction that passes through user $i$'s channel with the bank is $$\frac{w_i z \cdot (1+r)^{T(w_i)} +a \phi - w_i z}{T(w_i)\cdot \ell_i}.$$ The bank does not support any cost that user $i$ incurs, so the fee for user $i$ is twice as high compared to the case where $i$ has transactions with only one other user $j$ (with no intermediary).
\end{proof}

The analysis from the pairs topology applies here too, except that the cost of using the lightning network for each user is multiplied by a factor of two.

\section{Discussion}

Instead of Poisson transfers, we can assume that there is a fixed sequence of transfers: $(i_1,j_1)$, $(i_2,j_2),\ldots$, where transfer number $k$ is from user $i_k$ to user $j_k$. 
Each user $i$ has a fixed initial budget $K_i$.
 The sequence is feasible, i.e, at time $k$, user $i_k$ always has a sufficient amount of money for transferring to user $j_k$. 
 We assume that at time $k$, user $i$ considers only two options: blockchain transfer or bitcoin transfer, and selects the cheaper option.
There is an interesting algorithmic question: given a fixed sequence of transfers, what is the optimal configuration of lightning channels?

\section{Acknowledgments}

%
The symbolic calculations were done with Sympy \cite{sympy}, the 
plots with Desmos \cite{desmos} 
and 
Matplotlib \cite{matplotlib}.

\nocite{*}

\bibliographystyle{alpha}
\bibliography{bitcoin}

\newcommand{\etalchar}[1]{$^{#1}$}
\begin{thebibliography}{KYP{\etalchar{+}}21}

\bibitem[BCV20]{BCV20}
S.~Bartolucci, F.~Caccioli, and P.~Vivo.
\newblock A percolation model for the emergence of the bitcoin lightning
  network.
\newblock {\em Nature Scientific Reports}, 10, 2020.

\bibitem[BNT22]{BNT22}
A.~Biryukov, G.~Naumenko, and S.~Tikhomirov.
\newblock Analysis and probing of parallel channels in the lightning network.
\newblock In {\em Financial Cryptography}, 2022.

\bibitem[Bre73]{brent}
R.P. Brent.
\newblock {\em Algorithms for Minimization without Derivatives}.
\newblock Prentice-Hall, Englewood Cliffs, NJ, 1973.
\newblock Ch. 4.

\bibitem[BSB21]{Beres2021Cryptoeconomic}
F.~B{\' e}res, I.~A. Seres, and A.~Bencz{\' u}r.
\newblock A {Cryptoeconomic} {Traffic} {Analysis} of
  {Bitcoin}\textquoteright{}s {Lightning} {Network}.
\newblock {\em Cryptoeconomic Systems}, 0(1), apr 5 2021.

\bibitem[CDE{\etalchar{+}}16]{Croman2016}
K.~Croman, C.~Decker, I.~Eyal, A.~E. Gencer, A.~Juels, A.~Kosba, A.~Miller,
  P.~Saxena, E.~Shi, E.~G{\"u}n~Sirer, D.~Song, and R.~Wattenhofer.
\newblock On scaling decentralized blockchains.
\newblock In {\em Financial Cryptography and Data Security: FC 2016
  International Workshops}, pages 106--125. Springer Berlin Heidelberg, 2016.

\bibitem[des]{desmos}
Desmos.
\newblock http://www.desmos.com.

\bibitem[DW15]{DW15}
C.~Decker and R.~Wattenhofer.
\newblock A fast and scalable payment network with bitcoin duplex micropayment
  channels.
\newblock In {\em 17th International Symposium on Stabilization, Safety, and
  Security of Distributed Systems}, volume 9212, pages 3--18, 2015.

\bibitem[GHS21]{GGS21}
P.~Guasoni, G.~Huberman, and C.~Shikhelman.
\newblock Lightning network economics: Channels.
\newblock {\em Michael J. Brennan Irish Finance Working Paper Series Research
  Paper No. 21-7, Columbia Business School}, 2021.

\bibitem[HS]{HS17}
M.~Hearn and J.~Spilman.
\newblock Bitcoin contracts.
\newblock \url{https://en.bitcoin.it/wiki/Contracts}.

\bibitem[Hun07]{matplotlib}
J.~D. Hunter.
\newblock Matplotlib: A 2d graphics environment.
\newblock {\em Computing In Science \& Engineering}, 9(3):90--95, 2007.

\bibitem[KYP{\etalchar{+}}21]{Kappos_21}
G.~Kappos, H.~Yousaf, A.~Piotrowska, S.~Kanjalkar, S.~Delgado-Segura,
  A.~Miller, and S.~Meiklejohn.
\newblock An empirical analysis of privacy in the lightning network.
\newblock In {\em Financial Cryptography and Data Security}, pages 167--186.
  Springer Berlin Heidelberg, 2021.

\bibitem[LR06]{recurrence_lecture_notes}
Tom Leighton and Ronitt Rubinfeld.
\newblock Random walks.
\newblock https://web.mit.edu/neboat/Public/6.042/randomwalks.pdf, 2006.

\bibitem[NFSD20]{NFSD20}
U.~Nisslmueller, K.~Foerster, S.~Schmid, and C.~Decker.
\newblock Toward active and passive confidentiality attacks on cryptocurrency
  off-chain networks.
\newblock In {\em 6th International Conference on Information Systems Security
  and Privacy (ICISSP)}, 2020.

\bibitem[PD16]{PD16}
J.~Poon and T.~Dryja.
\newblock The bitcoin lightning network: Scalable off-chain instant payments,
  2016.
\newblock Technical report.

\bibitem[sym]{sympy}
Sympy.
\newblock https://www.sympy.org/en/index.html.

\bibitem[WZS{\etalchar{+}}22]{wang2022can}
Z.~Wang, R.~Zhang, Y.~Sun, H.~Ding, and Q.~Lv.
\newblock Can lightning network’s autopilot function use ba model as the
  underlying network?
\newblock {\em Frontiers in Physics}, 9:794160, 2022.

\bibitem[ZFDS22]{ZFDS22}
P.~Zabka, K-T. Foerster, C.~Decker, and S.~Schmid.
\newblock Short paper: A centrality analysis of the lightning network.
\newblock In {\em Financial Cryptography}, 2022.

\end{thebibliography}

\appendix

\section{Numeric examples} 
\label{sub:numeric-example}
In this section we illustrate our approach on an example with realistic numbers and observe the differences between a world with lightning and one without.

For the examples below we assume the following parameter values.
\begin{itemize}
\item Between each pair, there are $\ell = 10$ transfers per day. 
\if\asymmetric1
\item When the pair is asymmetric, we assume that $\Delta = 6$ per day, i.e, in average, Alice makes $8$ transfers to Bob and Bob makes $2$ transfers to Alice.
\else
\item The transfer-rate is symmetric, i.e, each of Alice and Bob makes an average of 5 transfers per day.
\fi
\item $\beta = 0.01$, i.e, the utility of a user from a money-transfer is $1\%$ of the transfer-size.
\item The interest rate is $4\%$ per year. This implies that $r \approx 4/365/100 \approx 0.0001096$ per day.
\item The size of a channel-reset transaction is $1.1$ times the size of a usual transaction, i.e, $a=1.1$ [records].
\item The minimum transfer-size is $\zmin=0.001$ [bitcoins]:
\end{itemize}
These parameters give the following threshold values of $z$
\if\asymmetric1
\begin{align*}
\text{asymmetric case:} 
&&
\text{symmetric case:} 
\\
\tNL \approx 0.29 \phi 
&&
\tNL \approx 0.0036 \phi 
\\
\tNB = 100 \phi
&&
\tNB = 100 \phi
\\
\tLB \approx 34564 \phi 
&&
\tLB \approx 16744 \phi 
\end{align*}
\else
(recall we are assuming symmetric channels):
\begin{itemize}
\item $\tNL = \left({27 a r^2 \over \ell^2\beta^3} \right) \cdot \phi \approx 0.0036 \phi $
\item $\tNB = \left( {1\over \beta}\right) \cdot \phi = 100 \phi$
\item $\tLB = \left( {\ell\over \sqrt{27 a r^2}}\right) \cdot  \phi \approx 16744 \phi$
\end{itemize}
\fi
So, without lightning, users use blockchain iff their transfer-size is above $100\phi$; with lightning, in the symmetric case:
\begin{itemize}
\item  Users whose transfer-size is less than $0.0036 \phi$, do not transfer at all;
\item Users whose transfer-size is between $0.0036 \phi$ and $16744 \phi$ use lightning; 
\item Users whose transfer-size is above $16744 \phi$, use direct blockchain transfers.
\end{itemize}

This has several interesting implications:
\begin{itemize}
\item In a world with lightning, whenever $\phi$ is even slightly above zero (e.g. $\phi > 10^{-7}$) almost all users user lightning --- the number of direct blockchain transfers is negligible.
\item 
Since  $\tNL\ll \tLB$,
the demand function for the range of high $\phi$ (e.g. $\phi>0.3$) is dominated by the factor:
\begin{align*}
3\zmin\sqrt[3]{\ell a r^2\over \phi^2} \tNL^{-1/3}
=
3\zmin\sqrt[3]{\ell^3 \beta^3\over 27\phi^3}
=
\ell{\zmin \over \tNB}
\end{align*}
which is the demand without lightning. So, for this $\phi$ we expect the demand curves with and without lightning to be almost overlapping.
\end{itemize}

\if\uniform1
\subsubsection{Uniformly-distributed transfer size}
Suppose $z$ is distributed uniformly between $0$ and $\zmax$, where $\zmax = 1$ [bitcoins]. Then:

Without lightning, the expected demand per pair is:
\begin{align*}
D^{0}(\phi) &=
\begin{cases} \ell - \frac{\ell \phi}{\beta z_{\max}} & \text{for}\: \phi < \beta z_{\max} \\0 & \text{otherwise} 
\end{cases}
\end{align*}

With lightning, in the asymmetric case, the expected demand per pair is:
\begin{align*}
D^a(\phi) =
\begin{cases} - \frac{\ell^{3} \phi}{6 \Delta a z_{\max} r} + \ell - \frac{16 \Delta^{2} a^{2} \phi r^{2}}{3 \ell^{3} \beta^{3} z_{\max}} & \text{for}\: \phi < \frac{4 \Delta}{\ell^{2}} a z_{\max} r \\\frac{2 \sqrt{\Delta} \sqrt{a} \sqrt{z_{\max}}}{3 \sqrt{\phi}} \sqrt{r} - \frac{16 \Delta^{2} a^{2} \phi r^{2}}{3 \ell^{3} \beta^{3} z_{\max}} & \text{for}\: \phi < \frac{\ell^{2} \beta^{2} z_{\max}}{4 \Delta a r} \\0 & \text{otherwise} \end{cases}
\end{align*}

In the symmetric case, the expected demand per pair is:
\begin{align*}
D^s(\phi) =
\begin{cases} - \frac{4 \sqrt{3} \ell^{2} \phi}{45 \sqrt{a} z_{\max} r} + \ell - \frac{729 a^{2} \phi r^{4}}{5 \ell^{3} \beta^{5} z_{\max}} & \text{for}\: \phi < \frac{3 z_{\max}}{\ell} \sqrt{3} \sqrt{a} r \\\frac{3 \sqrt[3]{\ell} \sqrt[3]{a} z_{\max}^{\frac{2}{3}}}{5 \phi^{\frac{2}{3}}} r^{\frac{2}{3}} - \frac{729 a^{2} \phi r^{4}}{5 \ell^{3} \beta^{5} z_{\max}} & \text{for}\: \phi < \frac{\ell^{2} \beta^{3} z_{\max}}{27 a r^{2}} \\0 & \text{otherwise} \end{cases}
\end{align*}
The demand curves for the above parameters are shown in Figure \ref{fig:demand-curves}.
Initially, the demand with lightning drops fast, but then it remains positive while the demand without lightning drops to zero.

\begin{figure}
\begin{center}
\includegraphics[scale=1]{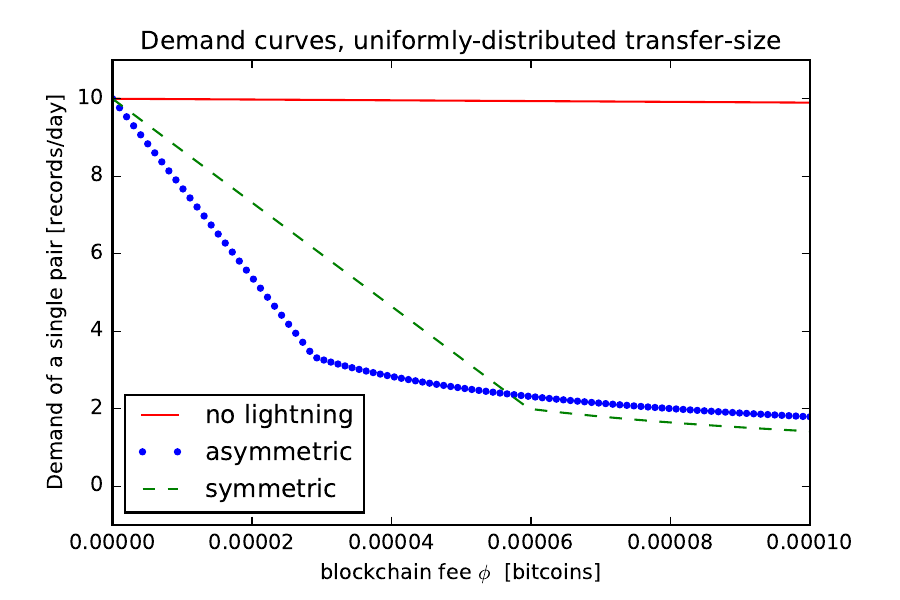}
\hskip 0.08 \textwidth
\includegraphics[scale=1]{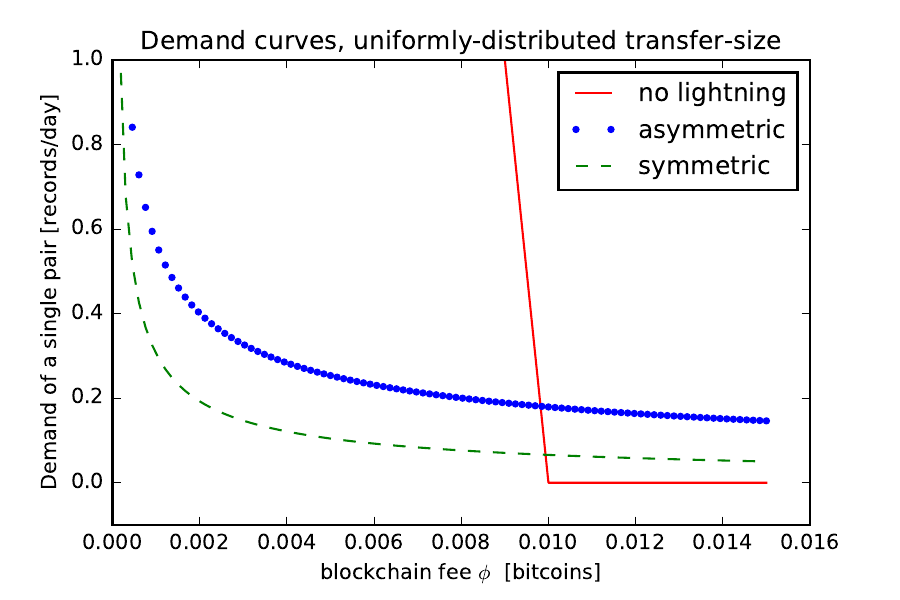}
\end{center}
\caption{
\label{fig:demand-curves}
Demand curves when $z$ is distributed uniformly. 
\\
The two panes show two different regions of the same curves.
\\
\textbf{Left:} price in $[0,0.0001]$.
The demand with lightning (both in the symmetric and asymmetric case) is much lower than the demand without lightning.
\\
\textbf{Right:} price in $[0,0.015]$.
When the blockchain fee increases, the demand without lightning quickly goes to 0 while the demand with lightning remains positive.
}
\end{figure}

Based on the demand curves, we calculate the market-equilibrium price as a function of $n$ --- the number of pairs (so the number of users is $2 n$).

Without lightning, the equilibrium price is:
\begin{align*}
\phi_{eq}^{0} = 
\begin{cases}
\beta \zmax (1 - {\tau\over n \ell})
& n \ell > \tau
\\
0 & n \ell \leq  \tau
\end{cases}
\end{align*}
With lightning, the equilibrium price is much harder to compute. I could compute it only in the asymmetric case, and even in that case, the expression is very complicated. It is plotted in Figure \ref{fig:price-curves}.

Interestingly, while without lightning the price converges at some point (to $\beta \zmax$ --- the maximum value of a transfer), with lightning the price increases indefinitely. This is because with lightning the demand is positive even when the price is very high.
\\
To calculate the total daily revenue of the blockchain miners, the market-equilibrium price $\phi$ should be multiplied by $\tau$, the number of records mined per day. Currently $\tau = 288,000$, so the miners' revenue converges at about $2900$ bitcoins/day without lightning and increases much higher with lightning.

\begin{figure}
\begin{center}
\includegraphics[scale=1]{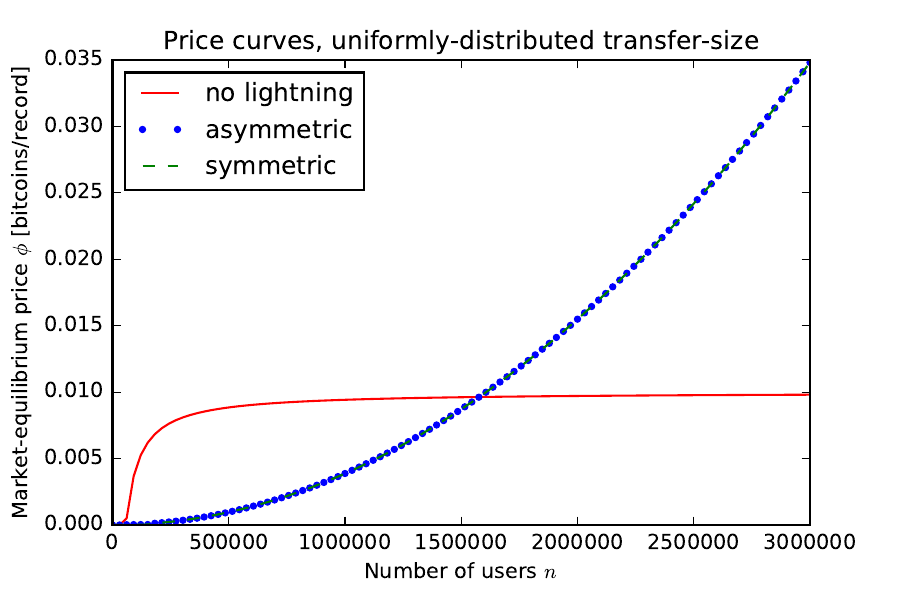}
\end{center}
\caption{
\label{fig:price-curves-uniform}
Price curves when the transfer-size is distributed uniformly.
}
\end{figure}

Finally, Figure \ref{fig:txs-curves-uniform}
shows the effect of the lightning network on the total number of transactions in the market. While without lightning the total number of transactions is upper-bounded by $\tau$ (horizontal solid line), with lightning the number of transactions can increase much higher. As the number of users increases, the blockchain-fee increases, so more and more users switch to using lightning. There is an interesting discontinuity in the graph, which we plan to investigate in future work.

\begin{figure}
\begin{center}
\includegraphics[scale=1]{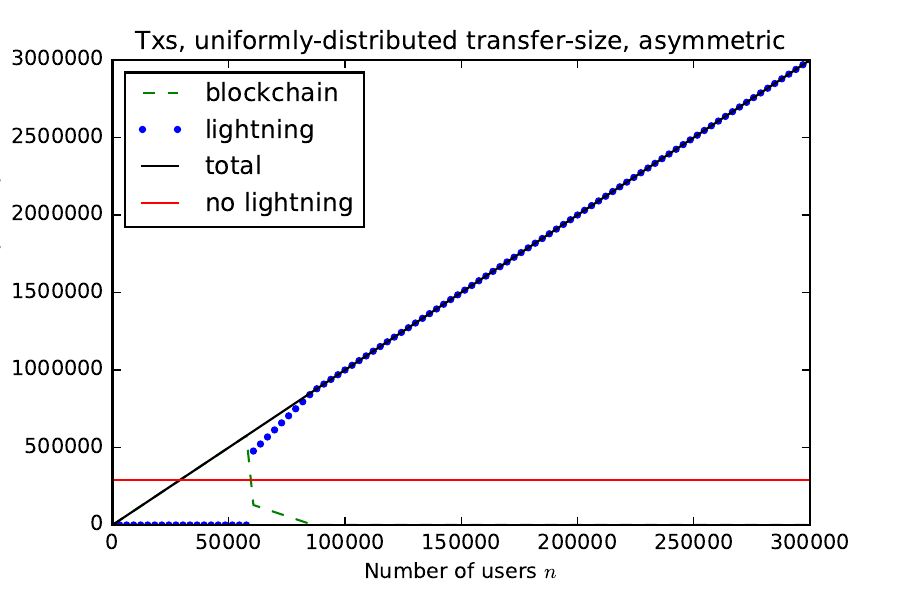}
~~
\includegraphics[scale=1]{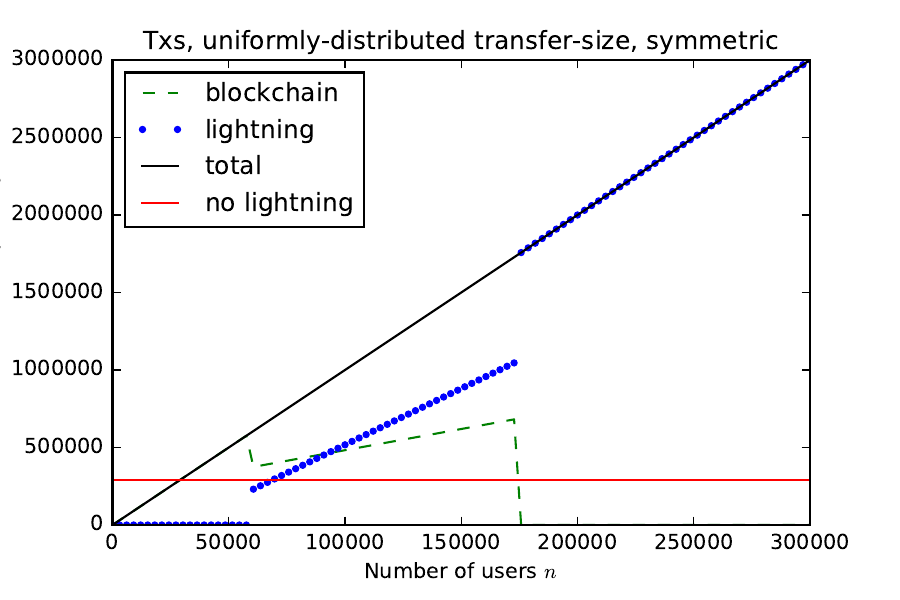}
\end{center}
\caption{
\label{fig:txs-curves-uniform}
Numbers of transactions when the transfer-size  is distributed uniformly.
}
\end{figure}
\fi 
\subsubsection{Demand curves}
Substituting in Corollary \ref{cor:demand-powerlaw}, the demand functions are:
\begin{align*}
D^{0}(\phi) &=
\begin{cases}
10 & [\phi < 10^{-5}]
\\
{10^{-4}\over\phi} & [10^{-5}\leq \phi]
\end{cases}
&& \text{(without lightning)}
\\
D^s(\phi) &= 
\begin{cases}
10 
& [\phi < 6\cdot 10^{-8}]
\\
{1.5\cdot 10^{-5}}
\left({10\over\phi^{2/3}} - {0.04\over\phi}\right) + {6\cdot 10^{-7}\over\phi} 
& [6\cdot 10^{-8} \leq \phi < 0.3]
\\
{1.5\cdot 10^{-5}}
\left({6.5\over\phi} - {0.04\over\phi}\right) + {6\cdot 10^{-7}\over\phi} 
& [0.3 \leq \phi]
\end{cases}
&& \text{(with lightning)}
\end{align*}
The transactions counts are: 
\begin{align*}
\begin{cases}
Lightning: 0, ~ ~ 
Blockchain: 10 
& [\phi < 6\cdot 10^{-8}]
\\
Lightning: \ell \left(1-{\zmin\over{\tLB}}\right), ~ ~ 
Blockchain: \ell \left({\zmin\over{\tLB}}\right)
& [6\cdot 10^{-8} \leq \phi < 0.3]
\\
Lightning: \ell \left({\zmin\over\tNL}-{\zmin\over{\tLB}}\right), ~ ~ 
Blockchain: \ell \left({\zmin\over{\tLB}}\right)
& [0.3 \leq \phi]
\end{cases}
\end{align*}


The demand curves for the above parameters are shown in Figure \ref{fig:demand-curves-powerlaw}.

\begin{figure}[h!]
\begin{center}
\fbox{
\includegraphics[scale=0.5]{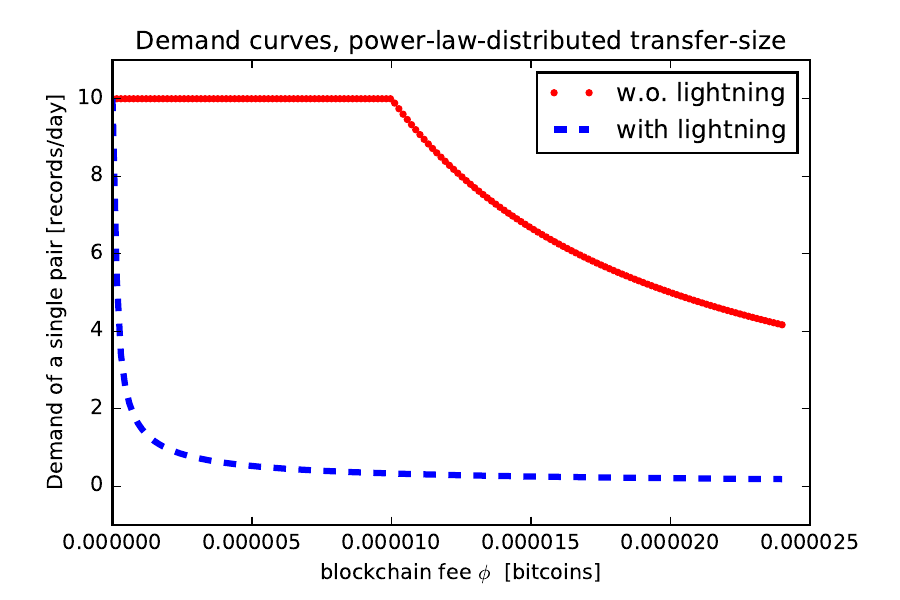}
~~
\includegraphics[scale=0.5]{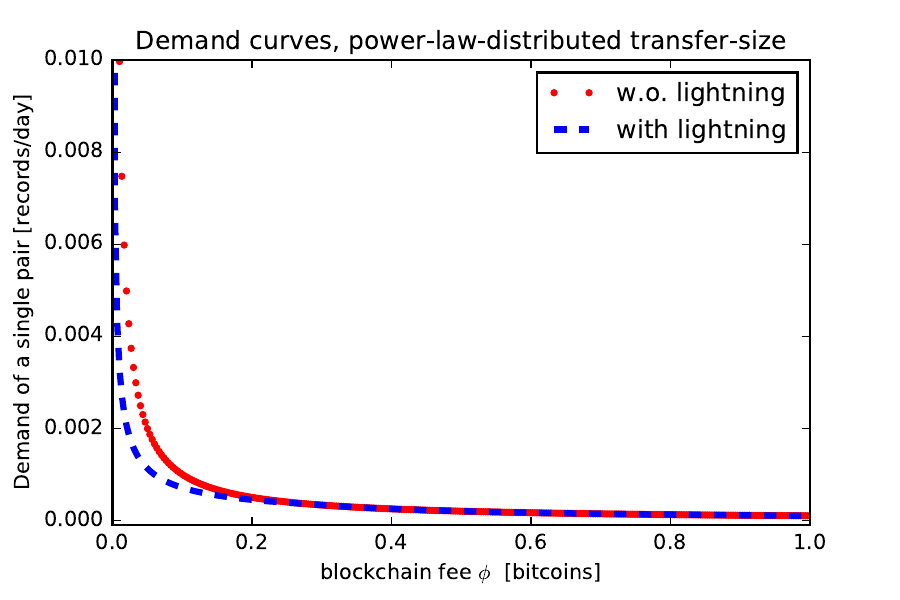}
}
\\
\fbox{
\includegraphics[scale=0.5]{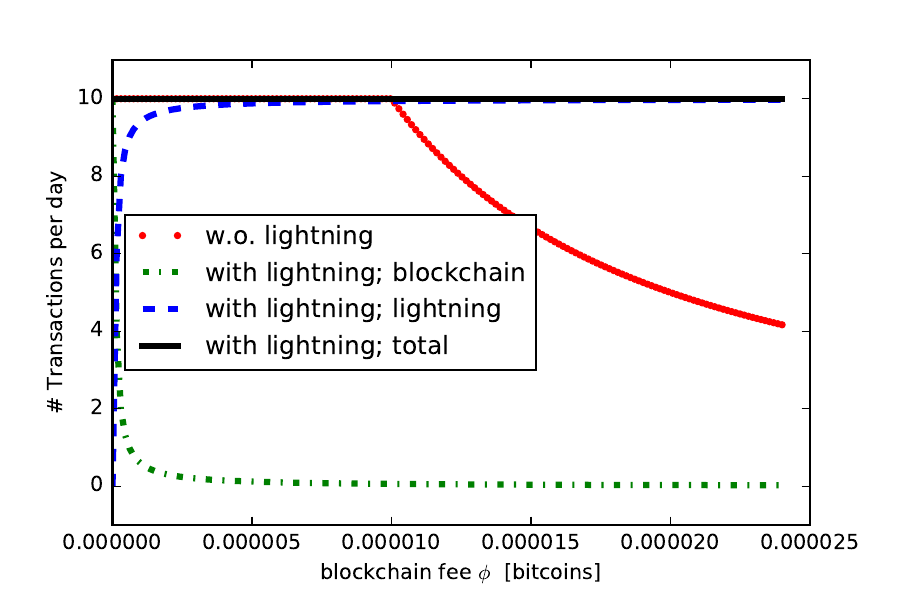}
~~
\includegraphics[scale=0.5]{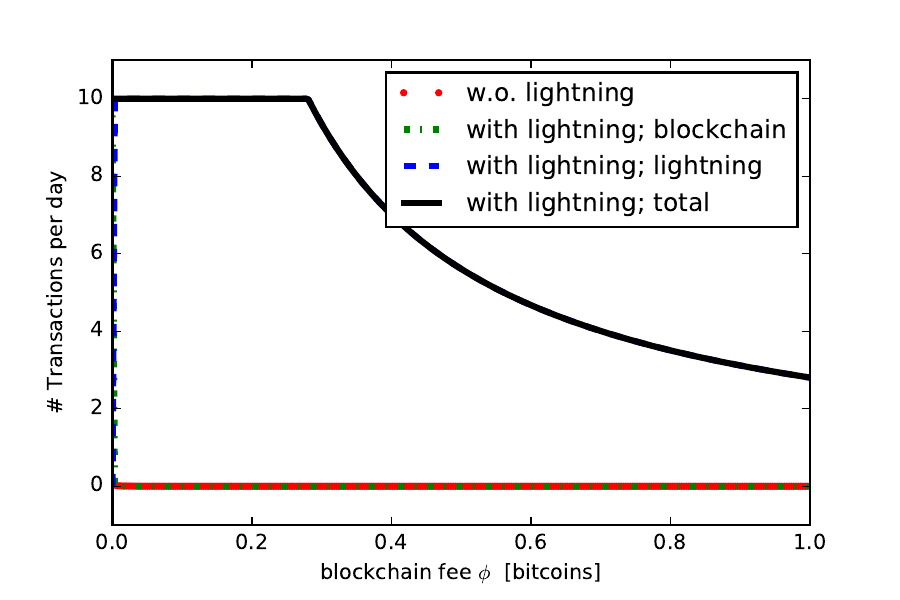}
}
\end{center}
\caption{
\label{fig:demand-curves-powerlaw}
Demand curves when $z$ has  
power-law distribution and the transfer-rate is symmetric.
\\
\textbf{Top:} Records-per-pair in a world with and without lightning (two different price-ranges).
\\
\textbf{Bottom:} Transfers-per-pair in a world with and without lightning  (two different price-ranges).
}
\end{figure}

As $\phi$ increases, we notice several interesting regions in the curves in Figure \ref{fig:demand-curves-powerlaw}:
\begin{itemize}
\item 
When $\phi$ is very low, all transfers above $\zmin$ are profitable, so all demand curves start at $\ell=10$ (the maximum demand per day).
In a world with lightning, a small increase in $\phi$ induces many users to prefer a lightning channel, so the demand for blockchain transfers declines quickly.
In a world without lightning, initially the demand remains high since users have no alternative, but then starts to decline as $\phi$ becomes too high for blockchain transfers (Left).
\item Both curves continue to decline with the price. In a world with lightning, almost all decline in demand comes from a decline in direct blockchain transfers; the number of lightning transfers remains at its maximum ($\ell$) even with $\phi$ becomes quite high. But when $\phi$ becomes very high ($\approx 0.3$), even the sparse channel-resets become too expensive. The two demand-curves meet and the number of lightning transfers starts to decline too (Right).
\end{itemize}

\subsubsection{Price curves}
\label{sub:price-curves}
Based on the demand curves and Theorem \ref{lem:equilibrium-fee}, we calculate the market-equilibrium price as a function of $n$ --- the number of agents. Note that the number of pairs is $n/2$.

The price is 0 as long as the maximum demand $\ell n/2$ is less than the fixed supply $\tau$. When the maximum demand hits the supply, the price starts to increase and the demand decreases.

Without lightning, the equilibrium price is given by:
\begin{align*}
\begin{cases}
{n\over 2} \ell\zmin {\beta \over \phi} = \tau
& [{\phi \over  \beta} < \zmin]
\end{cases}
\end{align*}
which leads to the following price-curve:
\begin{align*}
\phi^0_{eq} = 
\begin{cases}
0 & \text{for~}~ n /2 < \tau / \ell
\\
\zmin
\cdot 
\beta 
\ell 
\cdot
{n \over 2\tau}
& \text{otherwise}
\end{cases}
\end{align*}
Here, when the maximum demand hits the supply, the price jumps to $\beta \zmin$ (see Figure \ref{fig:price-curves-powerlaw}, top left); the discontinuity in the price-curve comes from the discontinuity in the power-law distribution. The equilibrium price then increases linearly with $n$ to keep the demand and supply equal (see Figure \ref{fig:price-curves-powerlaw}, top right).

With lightning, 
\if\asymmetric1
in the asymmetric case the equilibrium price is:
\begin{align*}
\phi^a_{eq} =
\begin{cases} 
0 & \text{for}\: \tau > \ell n 
\\\frac{4 \Delta}{\tau^{2}} a z_{\min} n^{2} r & \text{for}\: \tau > \frac{4 \Delta a n}{\ell \beta} r \\\frac{\ell \beta}{\tau} z_{\min} n & \text{otherwise} \end{cases}
\end{align*}
\fi
\if\symmetric1
in the symmetric case the equilibrium price is:
\begin{align*}
\phi^s_{eq} =
\begin{cases}
0
&
\text{when~}~ n /2 < \tau / \ell
\\
\zmin \cdot \sqrt{27 \ell a} \cdot r \cdot 
\left({n\over 2\tau}\right)^{3/2}
&
\text{when~}~ 
\tau/\ell < n/2 < {\beta^2 \ell \tau \over 27 a r^2}
\\
\zmin \cdot (\beta\ell - \sqrt{27 a} r)\cdot {n\over 2\tau}
&
\text{otherwise}
\end{cases}
\end{align*}
\fi
Here, when the maximum demand hits the supply,
the price jumps to $\zmin \sqrt{27 a r^2\over \ell^2}$, which is in general much lower than $\beta \zmin$ (see Figure \ref{fig:price-curves-powerlaw}, top left); even a small price-increase is sufficient to induce users to use lightning, thus reducing the demand for blockchain records. 
The price-increase is initially super-linear ($\Theta(n^{3/2})$, but eventually $\phi$ reaches its critical point where even channel-resets are too expensive ($\phi\approx 0.3$). At this point, the price starts to increase linearly with $n$ and the two price-curves coincide (see Figure \ref{fig:price-curves-powerlaw}, top right).

The price with lightning network, hence also the miners' revenue, is always below the price without lightning.

Below each price-curve 
we show the total number of transactions in the market and how they split between lightning and the blockchain:
\begin{itemize}
\item When the number of users is small. In a world without lightning, the number of transactions quickly attains its maximum of $\tau$. In a world with lightning, while the number of blockchain transactions decreases, a small number of such transactions can support a much larger number of lightning transactions (Left). 
\item When the number of users grows, almost all transactions are done in lightning (Right). Eventually, all $\tau$ transactions per day are used for reset transactions and the lightning network attains its maximum capacity, which is about $8e9$ --- more than 2 million times the blockchain capacity.
\end{itemize}

\begin{figure}[h!]
\begin{center}
\if\asymmetric1
\fbox{
\includegraphics[scale=1]{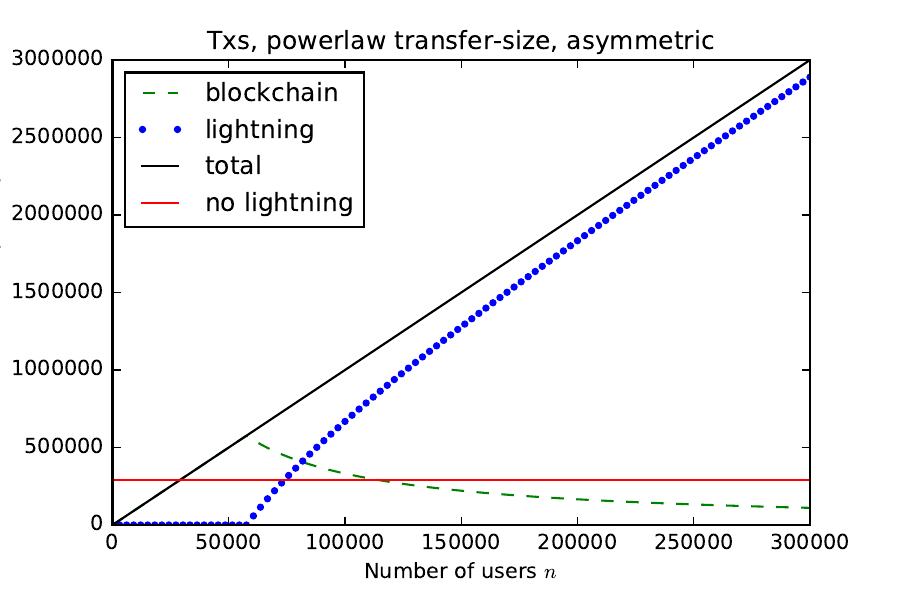}
}
~~
\fbox{
\includegraphics[scale=1]{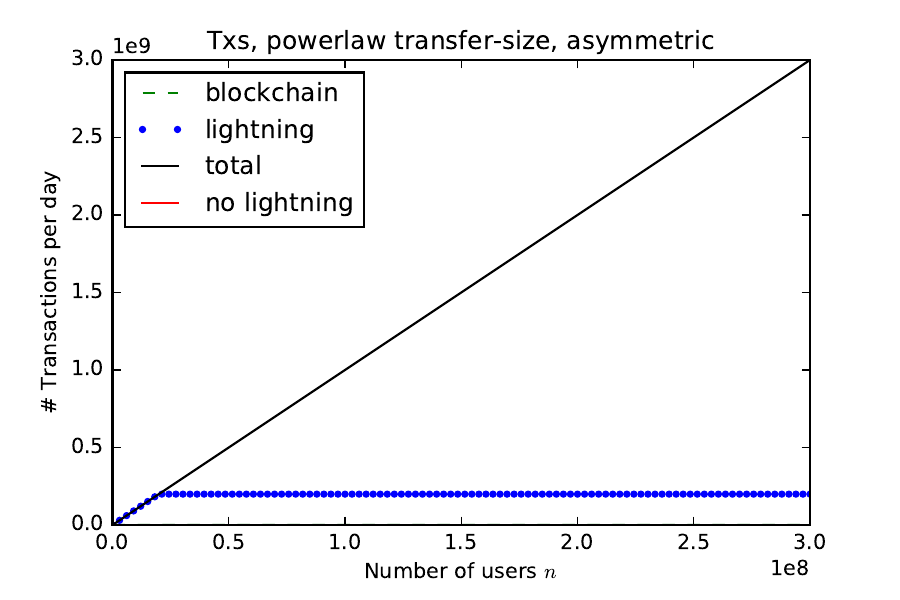}
}
\\
\fi
\if\symmetric1
\fbox{
\includegraphics[scale=0.5]{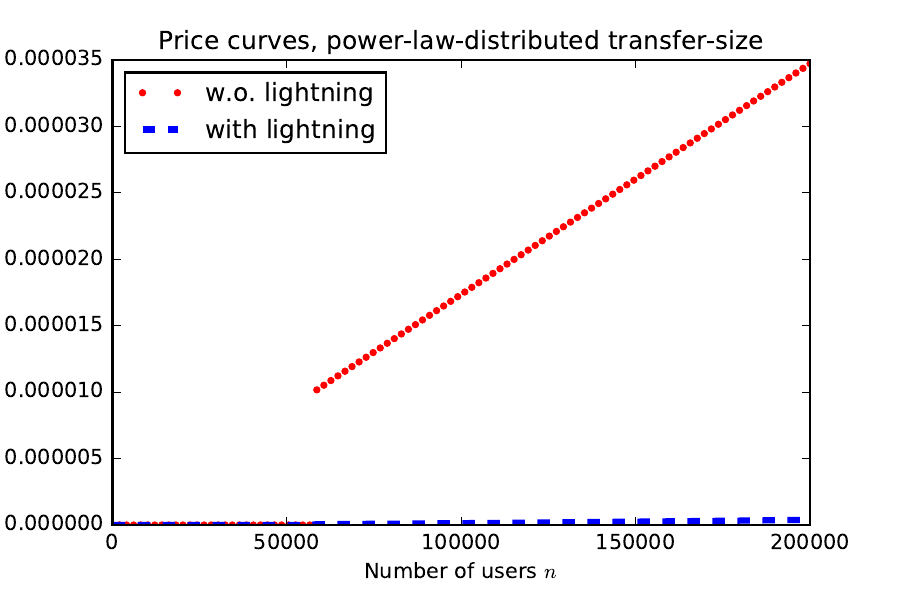}
~~
\includegraphics[scale=0.5]{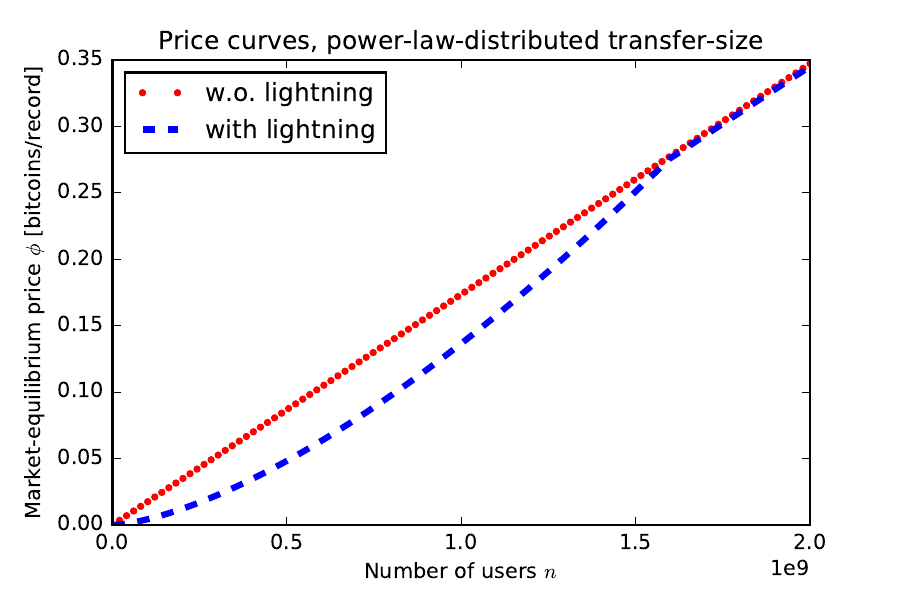}
}
\\
\fbox{
\includegraphics[scale=0.5]{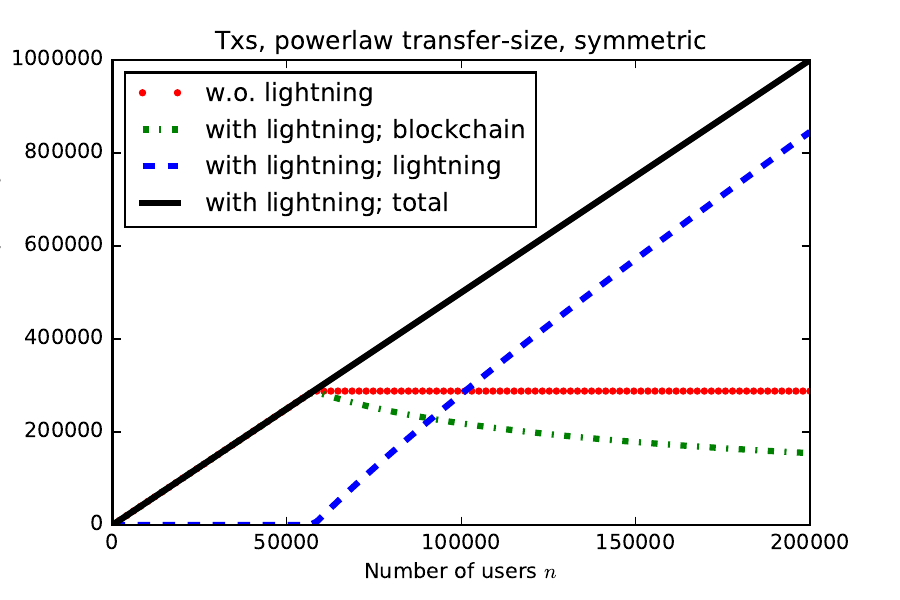}
~~
\includegraphics[scale=0.5]{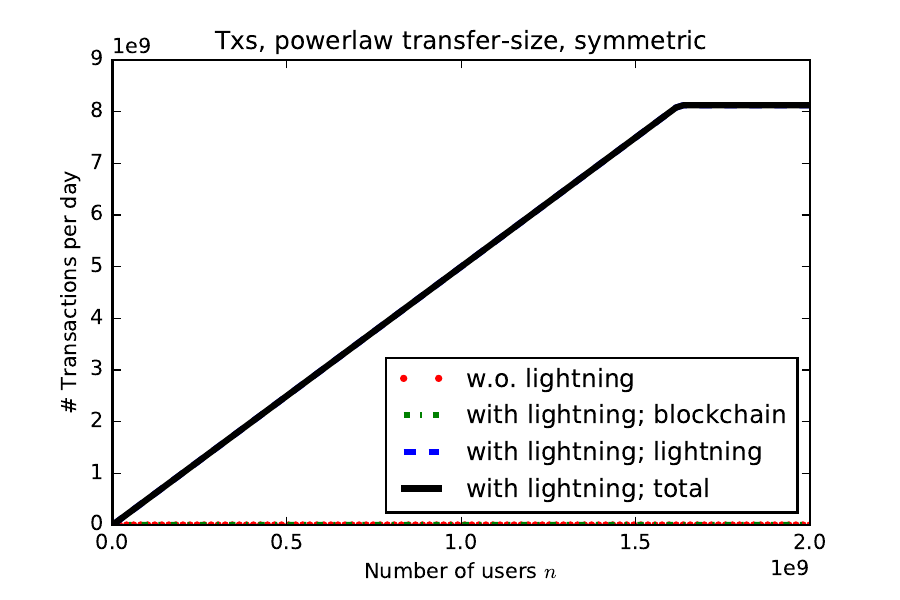}
}
\fi
\end{center}
\caption{
\label{fig:price-curves-powerlaw}
Price curves when $z$ has power-law distribution. 
\\
\textbf{Top:} Equilibrium price in a world with and without lightning (two different scales)
\\
\textbf{Bottom:} How the transfers split between lightning and blockchain  (two different scales).
}
\end{figure}

\end{document}